\documentclass[a4paper]{scrartcl}

\usepackage{cite}

\usepackage{amsmath}
\usepackage{amsfonts}
\usepackage{amssymb}
\usepackage{dsfont}
\usepackage{amsthm}
\usepackage{mathtools}% MoveEqLeft and other niceties MORTEN

\usepackage{lipsum} %Page numbering
\usepackage{bookmark}

\usepackage{authblk}

\usepackage[utf8x]{inputenc}
\usepackage[english]{babel}
\usepackage{german} 

\usepackage[final]{fixme}

\usepackage{graphicx}
\usepackage{subfig}

\makeatletter
\DeclareRobustCommand\widecheck[1]{{\mathpalette\@widecheck{#1}}}
\def\@widecheck#1#2{%
    \setbox\z@\hbox{\m@th$#1#2$}%
    \setbox\tw@\hbox{\m@th$#1%
       \widehat{%
          \vrule\@width\z@\@height\ht\z@
          \vrule\@height\z@\@width\wd\z@}$}%
    \dp\tw@-\ht\z@
    \@tempdima\ht\z@ \advance\@tempdima2\ht\tw@ \divide\@tempdima\thr@@
    \setbox\tw@\hbox{%
       \raise\@tempdima\hbox{\scalebox{1}[-1]{\lower\@tempdima\box
\tw@}}}%
    {\ooalign{\box\tw@ \cr \box\z@}}}
\makeatother

\DeclareMathOperator{\re}{Re}
\renewcommand{\Re}{\re}
\DeclareMathOperator{\im}{Im}
\renewcommand{\Im}{\im}
\newcommand{\bP}{\overline{P}} 
\newcommand{\inner}[3][]{#1\langle #2,#3 #1\rangle} % Inner products MORTEN

\newcommand{\tR}{{\widetilde{R}}}

\newcommand{\tC}{\widetilde{C}}
\newcommand{\hV}{\widehat{V}}
\newcommand{\wcV}{\widecheck{V}}

\newcommand{\mdsr}{\mathds{R}}
\newcommand{\mdsc}{\mathds{C}}

\newcommand{\inv}{^{-1}}
\newcommand{\ci}{\mathrm{i}}
\newcommand{\dom}{\mathcal{D}}

\newcommand{\Imp}[1]{\mathrm{Im}(#1)}

\theoremstyle{theorem}
\newtheorem{thm}{Theorem}[section]
\newtheorem{lem}[thm]{Lemma}

\newtheorem{prp}[thm]{Proposition}

\theoremstyle{definition}
\newtheorem{rem}[thm]{Remark}
\newtheorem{rems}[thm]{Remarks}
\newtheorem{defi}[thm]{Definition}
\newtheorem{nota}[thm]{Notation}
\newtheorem{ex}[thm]{Example}
\newtheorem{cond}[thm]{Condition}

\DeclareMathOperator{\Rank}{Rank}
\DeclareMathOperator{\Ran}{Range}
\DeclareMathOperator{\ad}{ad}
\providecommand{\inv}{^{-1}}

\newcommand{\R}{\mathds{R}}
\newcommand{\C}{\mathds{C}}

\newcommand{\N}{\mathds{N}}

\newcommand{\cH}{\mathcal{H}}
\newcommand{\cR}{\mathcal{R}}

\newcommand{\one}{\mathbf{1}}

\begin{document}

\bibliographystyle{amsplain}

\title{Local Spectral Deformation}%? Generalized Dilation Analyticity}

\author{Matthias Engelmann\footnote{Partially supported by the Lundbeck Foundation and the DFG via Graduiertenkolleg 1838}}

\affil{IADM\\ University of Stuttgart\\ Germany} 

\author{Jacob Schach M\o ller}

\affil{Department of Mathematics\\ Aarhus University\\ Denmark} 

% matthias@math.au.dk

\author{Morten Grud Rasmussen}

\affil{Department of Mathematical Sciences\\ Aalborg University\\ Denmark } % email: jacob@math.au.dk

\maketitle

--------------------------------------------------------------------------------------------------------------- 

\textbf{Abstract} We develop an analytic perturbation theory for eigenvalues with finite multiplicities, 
embedded into the essential spectrum of a self-adjoint operator $H$. We assume the existence of another self-adjoint operator
$A$ for which the family $H_\theta = e^{\ci\theta A} H e^{-\ci\theta A}$ extends analytically from the real line to a strip in the complex plane.
Assuming a Mourre estimate holds for $\ci[H,A]$ in the vicinity of the eigenvalue, we prove that the essential spectrum is locally deformed away from the eigenvalue, leaving it isolated and thus permitting an application of Kato's analytic perturbation theory.

---------------------------------------------------------------------------------------------------------------

\section{Introduction}

The investigation of the essential spectrum of a self-adjoint operator via spectral deformation techniques goes back to two papers by Aguilar--Combes and Balslev--Combes, see \cite{AgCo} and \cite{BaCo}. The starting point of the whole theory is the behavior of the Laplace operator under dilations. We define the unitary group of dilations on $\mathrm{L}^2(\R^d)$ by
\begin{equation*}
U(\theta)\psi(x) = \mathrm{e}^{\frac{d}{2}\theta} \psi(\mathrm{e}^{\theta}x) \quad \textup{for } \theta \in\R.
\end{equation*}
Under conjugation with $U(\theta)$ the Laplace operator transforms into
\begin{equation*}
U(\theta)\Delta U(\theta)\inv = \mathrm{e}^{-2\theta}\Delta.
\end{equation*}
Thus, the spectrum of $\mathrm{e}^{-2\theta}\Delta$ is a half-line starting at $0$ which has an angle of $-2\Imp{\theta}$ to the real line. 
%The situation can be interpreted in the following way: by conjugating $\Delta$ with $U(\theta)$ the spectrum of the transformed operator swings out into the complex plane with an angle dependent on size and sign of the imaginary part of $\theta$. 
The observation by Aguilar and Combes was that for certain one-body potentials $V$, the essential spectrum of the Schr\"odinger operator $H=-\Delta + V $ exhibits the same behavior, when conjugated with $U(\theta)$. 

This idea is generalized by Balslev and Combes to the situation of many-body Schr\"o\-ding\-er operators. After dilation, the essential spectrum consists of multiple half-lines, one starting at each threshold (eigenvalue of a subsystem Hamiltonian) protruding into the complex plane at a common angle $2\im \theta$.
Any non-threshold embedded eigenvalue will remain on the real axis, as an isolated eigenvalue of finite rank for which Kato's analytic perturbation theory applies \cite{K}.

 The class of (pair-)potentials for which this strategy works are called \emph{dilation analytic}. The theory of dilation analytic potentials and its application to quantum mechanics is summed up in \cite{RS4}. The method has been refined
to include potentials that may be locally singular using so-called exterior complex scaling, which is needed to treat e.g. Born-Oppenheimer molecules \cite{S_ext_scal}.

In the paper \cite{HS}, Hunziker and Sigal considered an abstract setup, where the unitary group $U(\theta)$
is, in principle, arbitrary and allowing for an analytic extension of $H_\theta = U(\theta) H U(\theta)^*$ into a strip around the real axis. Supposing that the continuous spectrum is locally deformed down into the lower half-plane, when $\im \theta>0$, leaving behind 
only isolated eigenvalues with finite rank Riesz projections, Hunziker and Sigal show that there is a one-one correspondence between embedded eigenvalues of $H$  and real eigenvalues of the deformed Hamiltonian $H_\theta$, in the region where the essential spectrum has been cleared away. This in turn permits an application of Kato's analytic perturbation theory for isolated eigenvalues of finite multiplicity, thus enabling an analytic perturbation theory of embedded eigenvalues as well as an analysis of resonances (poles of the resolvent are complex eigenvalues of the deformed Hamiltonian). 

In the present paper, we provide a natural set of  abstract conditions on a pair of self-adjoint operators $H$ and $A$ that ensures a local -- in energy -- deformation of the essential spectrum of $H_\theta$, leaving embedded eigenvalues isolated behind. Here
$A$ drives the unitary group $U(\theta) = \mathrm{e}^{\ci \theta A}$. Together with the results of \cite{HS}, this allows for an analytic perturbation theory of `non-threshold' embedded eigenvalues. In fact, exterior complex scaling may be viewed as an example of our general result. We note that there are refinements of exterior scaling that does not fit into our framework,  where $U(\theta)$ is not a group \cite{Hu_dist}.

To elucidate the role of the Mourre estimate in the theory of analytic deformations,
it is useful to expand $H_\theta$ as a formal power series:
\begin{equation}\label{eq:FormalSeries}
H_\theta = \mathrm{e}^{\ci\theta A}H\mathrm{e}^{-\ci\theta A}=
H - \theta \ci [H,A] + \frac{\theta^2}{2!} \ci^2[[H,A],A] - \cdots
\end{equation}
In fact, we shall in Subsect.~\ref{Subsec-GenDil} make sense out of this series strongly on $D(H)$.

Suppose $\lambda_0\in\R$ is an (embedded) eigenvalue of $H$.
Based on the expansion \eqref{eq:FormalSeries}, it is reasonable to expect that a Mourre estimate
\begin{equation}\label{IntroME}
\ci[H, A] \geq e -CE(|H-\lambda_0|\geq \kappa)\langle H\rangle - K
\end{equation}
will force the essential spectrum of $H_\theta$ with $\im \theta>0$ 
down into the lower half-plane, at least near $\lambda_0$. Here, as usual, $e,\kappa,C>0$ and $K$ is a compact operator.
The use of commutator estimates of this form goes back to Mourre \cite{Mo}.
 The compact error in the Mourre estimate leaves room for finitely many eigenvalues to stay behind. In fact, $\lambda_0$ will stay behind, but resonances -- eigenvalues with negative imaginary part --  may appear as well. Exploiting the Mourre estimate in conjunction with the series \eqref{eq:FormalSeries} is not new, cf. e.g. \cite{BGS,JMP}.

The main difficulty in establishing the spectral picture discussed in the preceding paragraph, comes from the fact that $H_\theta$ is not (in general) normal when $\im\theta \neq 0$. In Subsect.~\ref{Subsect-ME}, we assume a Mourre estimate and perform a Feshbach analysis to study the structure of the essential spectrum of $H_\theta$.  This puts us in a position to invoke \cite{HS}. In Subsect~\ref{Subsec-AnPert}, we employ Kato's analytic perturbation theory, to conclude a theorem on analytic dependence on parameters of embedded eigenvalues of $H$. 

The main result of this paper, Theorem~\ref{thm_ess_spec},  may be summed up succinctly as follows: Let $H,A$ be a pair 
of self-adjoint operators. Put $H_\theta = \mathrm{e}^{\ci\theta A} H  \mathrm{e}^{-\ci\theta A}$ for $\theta\in \R$, and assume
\begin{itemize}
\item $\forall \psi\in D(H)$, the map $\R\ni \theta\to H_\theta\psi$ (is well-defined and) extends to an analytic function in a strip around the real axis.
\item A Mourre estimate is satisfied for the pair $H,A$ in the vicinity, energetically, of an eigenvalue $\lambda_0$ of $H$.
\end{itemize}
Then, for $\theta$ with $\im \theta>0$ not too large, we have
\[
\sigma_\mathrm{ess}(H_\theta) \cap \bigl\{ z\in \C  \,  \big|\, \im z > - e'\im\theta/2, \ |\re z-\lambda_0| < \kappa' \bigr\} = \emptyset.
\]
Here $0<e'<e$ and $0<\kappa'<\kappa$ plays a role similar to $e$ and $\kappa$ in \eqref{IntroME}. Apart from the two main conditions itemized above, we have to impose some technical conditions on the pair $H$, $A$ in order for the analysis to go through. 

Combining our main result, Theorem~\ref{thm_ess_spec}, with Hunziker-Sigal \cite{HS} and Kato \cite{K},
yields an analytic perturbation theory for embedded `non-threshold' eigenvalues, summed up in Theorems~\ref{thm_kato}
and~\ref{thm_semianalytic}.

In Section~\ref{sec_ex}, we apply our analysis to two-body dispersive systems with real 
analytic one-body dispersion relations and a `dilation analytic' pair interaction. Such a system is translation invariant, and we study the analytic dependence of possible embedded non-threshold 
eigenvalues on total momentum.

The underlying motivation for this work in fact stems from three-body scattering for dispersive systems. 
While there are several unresolved issues surrounding scattering theory for three-body dispersive systems, one of them arises when dealing with scattering channels consisting of one incoming/outgoing free particle and one incoming/outgoing bound two-particle cluster. The free dynamics of the two-particle cluster is governed by an effective dispersion relation, which is in fact an eigenvalue of the two-body subsystem as a function of the total momentum of the two-particle cluster. If one cannot rule out the existence of embedded eigenvalues, then knowing that such effective dispersion relations are real analytic would allow one to argue that the associated threshold energies are nowhere dense. We can only say something about non-threshold energies, but that should in principle suffice, since the threshold set of two-body systems is well understood.

To conclude this introduction we discuss two examples. A trivial case of an operator which admits a band of embedded eigenvalues depending real-analytically on a parameter is provided in the following example. 
%It should be noted that we are able to treat this particular case with our methods, since it fits into a class of abstract examples treated in Section \ref{sec_ex}.

\begin{ex}[\!\!\!\cite{DerSasSand}] Let $H_0 =\Delta^2$ as an operator on $H^4(\R^d)$.
	Let $f\in C_\mathrm{0}^\infty(\R^d)$ be nonnegative, $f\geq 0$.
	
	Then, for any $\xi>0$  and since $(-\Delta + \xi)^{-1}$ is positivity improving,
	we have $u(\xi) = (-\Delta + \xi)^{-1}f$ is Schwartz class and strictly positive everywhere.
	
	Put 
	\begin{equation*}
	V(\xi) = -\frac1{u(\xi)} (-\Delta - \xi)f \in C_\mathrm{0}^\infty(\R^d).
	\end{equation*}
	Then
	\begin{equation*}
	V(\xi) u(\xi)  =  -(-\Delta - \xi)f = -(-\Delta - \xi)(-\Delta + \xi)u(\xi) = -\Delta^2 u(\xi) + \xi^2 u(\xi).
	\end{equation*}
	Hence $(H_0 + V(\xi))u(\xi) = \xi^2 u(\xi)$ and consequently, $H(\xi) = H_0 +V(\xi)$ has an embedded eigenvalue at the energy $E=\xi^2$.

	One can choose to read $V(\xi)$ as a function of $\xi>0$. The associated family of operators $H(\xi) = \Delta^2 + V(\xi)$ will now have a persistent (real analytic) 
	band of embedded eigenvalues $E(\xi) = \xi^2$. 
	%	Furthermore, $\xi \to H(\xi)$ is an analytic family of Type (A), in the sense of Kato \cite{Ka}.
\end{ex}

Since the strategy of the paper is to transform the Hamiltonian into a non-self-adjoint operator with receding essential spectrum in the area of interest, the question whether or not our assumptions are too strong arises. In particular, one could be tempted to hope that the minimal requirements of Kato's theory are sufficient. This however is not the case, since the following example illustrates that one cannot expect the usual conclusions of Kato to hold true, when one considers the behavior of embedded eigenvalues of self-adjoint operators under analytic perturbations. We recall from Kato \cite{K} that for one-parameter holomorphic families of self-adjoint operators, isolated eigenvalues of finite multiplicity may split up while locally preserving total multiplicity and forming real analytic branches that (suitably ordered) are real analytic through crossings. For non-normal holomorphic families it is only the algebraic multiplicity that is locally conserved (in $\mathbb C$), and eigenvalue branches may have at most algebraic singularities at crossings. 

\begin{ex}\label{ex-Simon} Let $\cH = L^2(\R^2)\oplus \C$ and define
\begin{equation*}
H(\xi) = \begin{pmatrix}
-\Delta - \xi^2 \one[|x|\leq 1] & 0 \\ 0 & 0
\end{pmatrix}
\end{equation*}
with domain $H^2(\R^2)\oplus \C$.
There exist $\rho>0$, such that For $\xi\in\R$ with $0 < |\xi| <\rho$, the operator  $-\Delta - \xi^2 \one[|x|\leq 1]$ has a unique eigenvalue $\lambda(\xi)$,
which is simple and depends real analytically on $0 < |\xi| < \rho$. See \cite{S_boun_dstates_2dim}. We may extend $\lambda$ to a continuous function on $(-\rho,\rho)$ by setting $\lambda(0)=0$. Hence, for $\xi\in (-\rho,\rho)$,  $\sigma_{\mathrm{pp}}(H(\xi)) = \{\lambda(\xi),0\}$ and $\xi\to H(\xi)$ is clearly analytic of Type (A).
We observe two things: \textbf{(I)} At $\xi=0$, there is a single simple eigenvalue $\lambda=0$. But we have two branches of eigenvalues coming out for $\xi\neq 0$. That is, the total multiplicity of the eigenvalue cluster is not upper semi-continuous. \textbf{(II)} The lower of the two eigenvalue branches $\xi\to \lambda(\xi)$ does not continue analytically through $\xi=0$, nor does it have an algebraic singularity at $\xi = 0$, more precisely; $|\lambda(\xi)|\leq e^{-(a\xi^2)^{-1}}$ for some $a>0$, cf.~\cite{S_boun_dstates_2dim}. For a closely related example, see \cite[p.~585]{GrieHas_analyt_pert}.
\end{ex}

 \noindent\textbf{Acknowledgement:} The authors thank an anonymous referee on an earlier version of this manuscript, who observed that one of our assumptions -- now removed -- from Section\ref{Subsect-ME} was superfluous.
 
\section{General Theory}\label{sec:GenTheory}

\subsection{Generalized Dilations}\label{Subsec-GenDil}

 In this subsection $H$ and $A$ denote two self-adjoint operators on a complex separable Hilbert
 space $\cH$. The inner product $\langle\cdot,\cdot\rangle$ is assumed linear in the second variable and
 conjugate linear in the first variable. We associate to an (unbounded) operator $T$ with domain $D(T)$, its graph norm $\|\psi\|_T = \|T\psi\| + \|\psi\|$, as a norm on the subspace $D(T)$. We shall frequently, for  self-adjoint $T$, exploit the easy estimate $\frac12 \|\psi\|_T\leq \|(T+\ci)\psi\|\leq \|\psi\|_T$.

  We work throughout Section~\ref{sec:GenTheory} under the following condition:

\begin{cond}\label{cond_g}
	\leavevmode
\begin{enumerate}
\item\label{cond_g_mourre}
 Abbreviating $U(t) = e^{\ci t A}$ for $t\in\R$, we assume
$U(t)D(H) \subset D(H)$ for all $t\in\R$ and
\begin{equation*}
\forall \psi\in D(H):\quad \sup_{|t|\leq 1}\|H U(t)\psi\|<\infty.
\end{equation*}
\item\label{cond_g_C1A}
 The quadratic form on $D(H)\cap D(A) \times D(H)\cap D(A)$ given by
\begin{equation*}
(\psi,\varphi)\mapsto \langle H\psi, A\varphi\rangle - \langle A\psi,H\varphi\rangle 
\end{equation*}
is continuous w.r.t. the  norm $\|(\psi,\varphi)\|_H := \|\psi\|_H + \|\varphi\|_H$.
\item\label{cond_g_ext}
 There exists $R>0$ such that for any $\psi\in D(H)$, the map
\begin{equation*}
\R \ni  t \mapsto H_{t}\psi := U(t)H U(-t)\psi
\end{equation*}
extends to a strongly analytic $\cH$-valued function $\{H_\theta\psi\}_{\theta\in S_R}$, where
\begin{equation}\label{StripDef}
S_R:=\bigl\{z\in\mdsc \,\big| \, |\mathrm{Im}(z)|<R\bigr\}.
\end{equation}
This defines a collection of linear operators $\{H_\theta\}_{\theta\in S_R}$ with domain $D(H)$.
\item\label{cond_g_uniform_bd} For $H_\theta$ defined above, note that $H_\theta(H+\ci)^{-1} \in \mathcal{B}(\mathcal{H})$ by the closed graph theorem.\footnote{$H_\theta$ is closable, since $H_{\bar{\theta}}\subset H_\theta^*$.} We suppose that %the map $U \rightarrow \mathcal{B}(\mathcal{H})$ given by $\xi \mapsto H_\theta(\xi)(H(\xi)+\ci)^{-1}$ extends to an analytic map on $U_{\C^d}$ and
\begin{equation*}
M:=\sup_{\theta \in B_R^{\C}(0)}\|H_\theta(H+\ci)\inv\|<\infty.
\end{equation*}
%\fxnote{Replace the symbol for the conjugation by the caligraphic one everywhere!}
\end{enumerate}
\end{cond}

\begin{rems} \label{rem:Cond_g}
\begin{enumerate}
\item\label{rem:C1Mourre} The Conditions \ref{cond_g}.\ref{cond_g_mourre} and  \ref{cond_g}.\ref{cond_g_C1A}
go back to Mourre \cite{Mo} and are equivalent to saying that $H$ is of class $C^1(A)$ with commutator $[H,A]^\circ$ bounded as an operator from $D(H)$ into $\cH$. See \cite[Prop.~B.11]{Mo_fully_coupled_PF}.
\item Another consequence of Conditions \ref{cond_g}.\ref{cond_g_mourre} and  \ref{cond_g}.\ref{cond_g_C1A} is the density of $D(H)\cap D(A)$
in both $D(H)$ and $D(A)$, equipped with their respective graph norms. See \cite[Lemma.~B.10]{Mo_fully_coupled_PF}.
\item\label{rem:ExtToDisc}
It suffices that the map $\theta\mapsto H_\theta\psi$ extends from $(-R,R)$ 
to $B^{\C}_R(0)$ in order to obtain an extension into $S_R$. Indeed, since we assume that $U(t) D(H) \subset D(H)$, the composition $H_\theta U(t)$ makes sense on $D(H)$ for all $\theta\in B_R^\C(0)$. Let $t\in \R$ and $\theta \in (t-R,t+R)$, then 
\begin{equation*}
H_{\theta}\psi=U(t)H_{\theta -t}U(-t)\psi
\end{equation*}
extends from $(t-R,t+R)$ to an analytic function on $B_R^\C(t)$ for all $\psi\in D(H)$. Sliding $t$ along the real axis produces an analytic continuation of $H_\theta\psi$ to the whole strip $S_R$.
\end{enumerate}
\end{rems}

We recall from \cite{K} that if $U\subset \C$ is open
then a family $\{T_\theta\}_{\theta\in U}$ of closed operators
is said to be \emph{analytic of Type (A)} if the domain of $T_\theta$ does not depend on $\theta$ and
the map $U\ni \theta\to T_\theta \psi$ is analytic for any $\psi$ in the common domain. If $U\subset \C^d$, then  $\{T_\theta\}_{\theta\in U}$ is said to be analytic of Type (A), if it is separately analytic of Type (A) in each of its $d$ variables.

\begin{lem}\label{lem:AnalDense} Assume Conditions \ref{cond_g}.\ref{cond_g_mourre} and  \ref{cond_g}.\ref{cond_g_C1A}. 
The following holds:
\begin{enumerate}
\item\label{item:AnalDense-1} For any $\psi\in D(H)$, $\theta\in\C$ and $m\in\N$, we have $\psi_m(\theta) := e^{- A^2/(2m) + \ci \theta A} \psi \in D(H)$.
\item\label{item:AnalDense-2} If $\theta \in \R$, we have
$\lim_{m\to \infty} \psi_m(\theta) = U(\theta)\psi$ in the topology of $D(H)$. In particular ($\theta=0$), the set of vectors in $D(H)$ that are analytic vectors for $A$ are dense in $D(H)$.
\item\label{item:AnalDense-3} For all $\psi\in D(H)$ and $m\in\N$, the map $\theta \to H \psi_m(\theta)$ is entire.
\end{enumerate}
\end{lem}

\begin{proof}Put $\psi_m(\theta) = e^{-A^2/(2m) + \ci \theta A}\psi$. Using the Fourier transform, we may write
\begin{equation*}
\psi_m(\theta) = \sqrt{\frac{m}{2\pi}} e^{-m\theta^2/2}\int_{\R} e^{- m t^2/2 + m \theta t}  U(t)\psi\, dt.
\end{equation*}
Note that for any $m\in \N$, the integral converges absolutely in $D(H)$, since $\|U(t)\psi\|_{D(H)} \leq e^{c|t|}$, for all $t\in\R$, where $c>0$ is some constant.
This is a consequence of  Condition~\ref{cond_g}.\ref{cond_g_mourre} and implies \ref{item:AnalDense-1}.

Let $\theta\in\R$.
To show that $\psi_m(\theta)\to U(\theta)\psi$ in $D(H)$, it suffices to argue that $H\psi_m(\theta)\to HU(\theta)\psi$ in $\cH$. Since $U(\theta)\psi\in D(H)$ for real $\theta$, it suffices to prove this with $\theta=0$.
Here we observe that 
\begin{equation*}
\sqrt{m} \int_{|t|\geq 1}e^{- m t^2/2}  U(t)\psi\, dt \to 0,
\end{equation*}
due to the estimate  $\|U(t)\psi\|_{D(H)} \leq e^{c|t|}$ from before. Furthermore, the estimate
\begin{equation}
\|(HU(t)- U(t) H)\psi\| = \|(H_{-t}-H)\psi\|\leq C|t|,
\end{equation}
valid for $|t|\leq 1$ with some $C>0$, follows from Condition~\ref{cond_g}\ref{cond_g_ext} and finally yields \ref{item:AnalDense-2}. 

We now establish  \ref{item:AnalDense-3}.
Since the map $\theta\to \psi_m(\theta)$ is entire it suffices, by Vitali-Porter's theorem, to show that
$n\| H (H+\ci n)^{-1} \psi_m(\theta)\|$ is bounded locally uniformly in $\theta\in\C$. But this follows easily from the estimates already invoked above.
\end{proof}

It turns out that under the assumption in Condition \ref{cond_g}.\ref{cond_g_mourre}, the remaining three items
are equivalent to the statement that all iterated commutators of $H$ with $A$ are $H$-bounded and satisfy a certain growth bound. If these bounds are satisfied the analytic continuation of the family $H_\theta$ can be written as a power series in a neighborhood of $0$. More precisely, we can prove

\begin{prp}\label{prp_equiv}
Assume Condition \ref{cond_g}.\ref{cond_g_mourre}. Then the following two properties are equivalent:
\begin{enumerate}
\item\label{item:equiv-cond-1}  Conditions \ref{cond_g}.\ref{cond_g_C1A}, \ref{cond_g}.\ref{cond_g_ext} and \ref{cond_g}.\ref{cond_g_uniform_bd}.
\item\label{item:equiv-cond-2}  There exists a constant $C>0$ such that: the iterated commutators $\ad_{A}^k(H)$ exist as $H$-bounded operators for all $k\in\mathds{N}$ and 
\begin{equation}
 \bigl\|\ad_{A}^k(H)(H+\mathrm{i})^{-1}\bigr\| \leq  C^k k!.
\label{eq_GenCommEst} 
\end{equation}
\end{enumerate} 
In the confirming case, $\{H_\theta\}_{\theta\in B_{(3C)^{-1}}^\C(0)}$ with common domain $D(H)$ is an analytic family of Type (A),
and for all $\theta\in B^\C_{(3C)^{-1}}(0)$ and $\psi\in D(H)$, we have 
\begin{equation}\label{PowerSeries}
H_{\theta}\psi:=\sum_{k=0}^{\infty}\frac{(-\theta)^k}{k!}\ci^k\ad_{A}^k(H)\psi
\end{equation}
and 
\begin{equation}\label{RelBounds}
 \frac12 \|\psi\|_{H}\leq \|\psi\|_{H_\theta}\leq 2 \|\psi\|_{H}.
 \end{equation}
\end{prp} 

\begin{rem} If one supposes \ref{item:equiv-cond-1} with given $R$ and $M$ coming from
 Condition \ref{cond_g}.\ref{cond_g_ext} and \ref{cond_g}.\ref{cond_g_uniform_bd}, respectively,  
 then one may choose $C = \max\{1,M\}/R$  in \eqref{eq_GenCommEst}. 
 
 Conversely, if one assumes \ref{item:equiv-cond-2} with a given $C$, then one may choose
$R= (3C)^{-1}$ and $M = 3$. 

Since we have elected to state our assumptions in terms of an analytic extension of $H$, we shall below employ
the estimate  \eqref{eq_GenCommEst} with  
\begin{equation}\label{eq:ChoiceOfC}
C = \frac{\max\{1,M\}}{R}.
\end{equation}
The expansion \eqref{PowerSeries} of $H_\theta$ and the relative bounds \eqref{RelBounds} will then hold true for $\theta\in B_{R'}^\C(0)$, where
\begin{equation}\label{eq:ChoiceOfRp}
R' = \frac1{3C} = \frac{R}{3 \max\{1,M\}}.
\end{equation}
\end{rem}

\begin{proof}  We begin with \ref{item:equiv-cond-2} $\Rightarrow$
\ref{item:equiv-cond-1}. Therefore, we assume that for all $k$, the iterated commutators exist as $H$-bounded operators $\ad_A^k(H)$ and that (\ref{eq_GenCommEst}) holds.

That  Condition \ref{cond_g}.\ref{cond_g_C1A} follows is obvious (take $k=1$).
 
Note that Condition \ref{cond_g}.\ref{cond_g_mourre} ensures that $H_\theta$ is well-defined for real $\theta$ 
as an operator with domain $D(H)$.

Exploiting (\ref{eq_GenCommEst}), we may for $\psi\in D(H)$  and $|\theta|< 1/C$ estimate
\begin{equation}
\sum_{k=0}^{\infty} \Bigl\|\frac{\theta^k}{k!}\ad_{A}^k(H)\psi\Bigr\|  \leq  \frac{\|(H+\ci)\psi\|}{1-C|\theta|}.
\label{eq_S_anal}
\end{equation}
%where $\|\cdot\|_{H}$ denotes the graph norm of $H$. 
Hence, the prescription
\begin{equation}
S_{\theta}\psi:=\sum_{k=0}^{\infty}\frac{(-\theta)^k}{k!}\ci^k\ad_{A}^k(H)\psi
\end{equation}
defines an analytic $\cH$-valued function defined in the disc $B^\C_{1/C}(0)$.
It is now easy to check that the map $\psi \mapsto S_{\theta}\psi$ defines -- for each $\theta \in B^\C_{1/C}(0)$ -- a linear operator with domain $D(H)$. 

The estimate \eqref{eq_S_anal} implies that
\begin{equation}
\forall \psi\in D(H):\quad \bigl\| S_\theta\psi\bigr\|  \leq  (1-C|\theta|)^{-1}\|\psi\|_{H},
\label{eq_S_anal2}
\end{equation}
and in particular that $S_\theta$ is $H$-bounded.

We proceed to show that $H$ and $S_\theta$ (for $\theta$ in a sufficiently small disc centered at $0$) 
define equivalent graph norms on $D(H)$. Note that \eqref{eq_S_anal2} already establishes that there exists a constant $C_1>0$ independent of $\theta \in B_{1/(3C)}^{\mathds{C}}(0)$ such that
\begin{equation*}
\|\psi\|_{S_{\theta}}\leq C_1 \|\psi\|_{H}.
\end{equation*} 
In complete analogy to the first estimate, we estimate for $\psi\in D(H)$ and  $\theta\in  B_{1/C}^{\C}(0)$:
\begin{align}
\|\psi\|_{H} 
&= 
\|\psi\| + \|H\psi\| \leq \|\psi\|_{S_{\theta}} + \sum_{k=1}^{\infty} (C|\theta|)^k\|\psi\|_{H} \nonumber\\
&=
\|\psi\|_{S_{\theta}} + \frac{ C|\theta|}{1-C|\theta|}\|\psi\|_{H}.
\label{eq_S-H_bd}
\end{align}
Hence, for $\theta \in  B_{1/(3C)}^{\mathds{C}}(0)$ we have
\begin{equation*}
 \|\psi\|_{H} \leq 2\|\psi\|_{S_{\theta}}.
\end{equation*}
This proves the claimed equivalence of graph norms and thus that $S_{\theta}$ is closed as an operator with domain $D(H)$ for all $\theta \in B_{1/(3 C)}^{\mathds{C}}(0)$.  Abbreviating $R= 1/(3C)$, we have now proved that $\{S_\theta\}_{\theta \in B_{R}^{\mathds{C}}(0)}$ is an analytic family of Type (A).
(Note that redoing the estimate \eqref{eq_S_anal} using $|\theta|\leq 1/(3C)$ yields $\|\psi\|_{S_\theta}\leq 2 \|\psi\|_H$ as well.)

It remains, recalling Remark~\ref{rem:Cond_g}.\ref{rem:ExtToDisc}, to argue that $S_\theta = H_\theta$ for $\theta\in (-R,R)$. 
Let  $\psi,\phi\in D(H)$ and put $\psi_m = e^{-A^2/(2m)}\psi$ and $\phi_m = e^{-A^2/(2m)}\phi$. Then, with the notation of Lemma~\ref{lem:AnalDense}, we have
\begin{equation*}
f_m(\theta) = \langle \psi_m, H_\theta \phi_m\rangle = \langle \psi_m(\bar{\theta}), H \phi_m(\theta)\rangle
\end{equation*}
a priori for real $\theta$, but extending to an entire function of $\theta$. Here we used
Lemma~\ref{lem:AnalDense}.\ref{item:AnalDense-3}.

We may use the assumption on the existence of iterated $H$-bounded commutators $\ad_A^k(H)$ to compute
\begin{equation*}
\frac{d^k f_m}{d\theta^k}\large{|}_{\theta=0} = \langle \psi_m, (-\ci)^k \ad_A^k(H)\phi_m\rangle.
\end{equation*}
Since analytic functions in $B_R^\C(0)$ are determined by their derivatives at zero, we may conclude that
\begin{equation*}
 \langle \psi_m, H_\theta \phi_m\rangle 
 =  \langle \psi_m, S_\theta \phi_m\rangle
\end{equation*}
for all $\theta\in B_R^\C(0)$. Finally, we exploit Lemma~\ref{lem:AnalDense} once more 
to compute the limit $m\to\infty$ in the above identity and conclude that
for all $\theta\in (-R,R)$ and $\psi,\phi\in D(H)\cap D(A)$, we have  $\langle \psi, H_\theta \phi\rangle 
  =  \langle \psi, S_\theta \phi\rangle$. By density of $D(H)\cap D(A)$ in $D(H)$, we conclude that $H_\theta = S_\theta$ for $\theta\in (-R,R)$ as desired. It now follows from \eqref{eq_S_anal2} that we may choose $M= 3$ in Condition \ref{cond_g}.\ref{cond_g_uniform_bd}.

%First note that $D(A)\cap\dom$ is a dense subset of $\dom$, since all $H(\xi)$ are assumed to be of class $\c1a$. Indeed, note %that this assumption implies that $(H(\xi)-z)\inv \in \c1a$ and that $(H(\xi)-z)\inv D(A)$ is a core for $A$ which is %contained in $D(A)\cap\dom$.

In order to prove that \ref{item:equiv-cond-1} $\Rightarrow$
\ref{item:equiv-cond-2}, we assume that 
 Conditions \ref{cond_g}.\ref{cond_g_C1A}--\ref{cond_g}.\ref{cond_g_uniform_bd} holds true.
 Let $\eta,\psi \in D(H)$. By Condition \ref{cond_g}.\ref{cond_g_mourre} and the analyticity of $\theta \mapsto H_{\theta}\psi$, we may use \cite[Prop.~2.2]{MoWe} to argue that all iterated commutators of $A$ with $H$ exists and are implemented by $H$-bounded operators, provided we can establish that for every $j\in\mathds{N}$ there exist $H$-bounded operators $H_0^{(j)}$, such that 
\begin{equation*}
\forall \theta \in (-R,R):\quad \frac{\mathrm{d}^j}{\mathrm{d}\theta^j} \langle \eta, H_{\theta}\psi\rangle \large{|}_{\theta=0} = \langle\eta , H_{0}^{(j)}\psi\rangle.
\end{equation*}
As a starting point we use the analyticity of $\theta \mapsto H_{\theta}\psi$ to obtain a power series expansion for $|\theta|<r < R$, that is
\begin{equation}
\langle \eta, H_{\theta} \psi \rangle= \sum_{k=0}^{\infty}\theta^k b_k(\eta,\psi),\quad b_k(\eta,\psi)= \frac{1}{2\pi\ci}\int\limits_{\Gamma_r} \theta^{-k-1}\langle \eta, H_{\theta}\psi \rangle\, \mathrm{d}\theta,
\label{eq_PwrSer}
\end{equation}
where $\eta\in\mathcal{H}$ and $\Gamma_r$ is the circle in the complex plane with radius $r$ centered at 0. Observe that
the $b_k(\eta,\psi)$'s define sesquilinear forms.

Using Condition \ref{cond_g}.\ref{cond_g_uniform_bd}, we get an $M>0$ such that
\begin{equation*}
|b_k(\eta,\psi)|\leq \|\eta\|\|\psi\|_{H}\frac{M}{R^{k}},
\end{equation*} 
where we also took the limit $r\to R$.
For every $\psi\in D(H)$ (and $k\in\N$) there thus exists a vector $\tilde{\psi}$ such that $b_k(\eta,\psi)=\langle \eta, \tilde{\psi}\rangle$ for all $\eta\in D(H)$. It follows that the assignment $B_k\psi:=\tilde{\psi}$ defines an $H$-bounded 
linear operator on $D(H)$.
With this construction, we have
\begin{equation*}
\frac{\mathrm{d}^j}{\mathrm{d}\theta^j} \langle \eta, H_{\theta}\psi\rangle \large{|}_{\theta=0} = \langle \eta, k! B_k\psi\rangle
\end{equation*}
and \cite[Prop.~2.2]{MoWe} now implies that \eqref{eq_GenCommEst} holds with $C:=\max\{1, M\}/R$. 
\end{proof}

In the following we abbreviate
\begin{equation}\label{eq_W}
W_\theta:= H_\theta-H = \sum_{k=1}^{\infty}\frac{(-\theta)^k}{k!}\ci^k\ad_A^k(H)
\end{equation}
as an operator with domain $D(H)$. Observe for $\theta \in B_{R'}^\C(0)$ the estimate
\begin{equation}\label{eq_W_est}
\|W_\theta(H+\ci)^{-1}\|\leq \frac{C|\theta|}{1-C|\theta|}\leq \frac{3 C}{2} |\theta|,
\end{equation}

We have the following -- rough but sufficient -- spectral localization result.

\begin{prp}\label{prp_spec}
\leavevmode
Assume Condition \ref{cond_g}. Then
\begin{equation*}
\forall \theta\in B_{R'}^\C(0):\quad \sigma(H_\theta) \subset \bigl\{ x+\ci y  \, \big| \, |y| \leq 4C |\theta|(|x|+1) \bigr\}.
\end{equation*}
\end{prp}

\begin{proof} Let $z\in \C$ with $\im z\neq 0$ and compute on $D(H)$:
\begin{equation*}
H_\theta -z = \left[1+ W_\theta(H-z)^{-1}\right](H-z)
\end{equation*}
Hence, $H_\theta-z$ is invertible if 
$\|W_\theta(H-z)^{-1}\|<1$
due to the Neumann series. The norm appearing in the previous inequality can be estimated trivially by
\begin{equation*}
\|W_\theta(H-z)^{-1}\|\leq \|W_\theta (H+\ci)^{-1}\| \sup_{p\in\R}\frac{|p+\ci|}{|p-z|}.
\end{equation*}
Let $c>0$.
Suppose $z = x+\ci y$ with $|y| \geq c(|x|+ 1)$. Then $|p+\ci|^2/|p-z|^2 \leq (p^2+1)/((p-x)^2+ c^2 x^2+ c^2) \leq 4/c^2$ uniformly in $p$, $x$ and $y$. 
Using \eqref{eq_W_est}, we have:
\begin{equation*}
\|W_\theta(H-z)^{-1}\|\leq  \frac{3 C|\theta|}{c} .
\end{equation*}
for $z=x+\ci y$ with $|y|\geq c|x|$
The choice $c = 4 C |\theta|$ ensures convergence of the Neumann series.
\end{proof}

\begin{lem}\label{lem_adjoint}
Assume Condition \ref{cond_g} and let $\theta\in B_{R'}^\C(0)$. We have
\begin{equation*}
D(H_\theta^*)= D(H) \quad \textup{and} \quad H_\theta^* = H_{\overline{\theta}}.
\end{equation*}
\end{lem}

\begin{proof}
Let $\psi, \phi\in D(H)$. We compute 
\begin{align*}
\bigl\langle \psi, H_\theta \phi\bigr\rangle 
&= 
\sum_{k=0}^{\infty}\Bigl\langle \psi, \frac{(-\theta)^k}{k!}\ci^k \ad_{A}^k(H)\phi\Bigr\rangle \\
&=
 \sum_{k=0}^{\infty} \Bigl\langle
 \frac{(-\overline{\theta})^k}{k!}\ci^k\ad_{A}^k(H)\psi, \phi\Bigr\rangle = \bigl\langle H_{\overline{\theta}}\psi,  \phi\bigr\rangle .
\end{align*}
Hence $H_{\overline{\theta}}\subset H_\theta^*$. Conversely, let $\phi\in D(H)$, $\psi \in D(H_\theta^*)$ and set
\begin{equation}\label{eq:choiceof-y}
y =  \max\bigl\{1, 8 C R'\bigr\}.
\end{equation} 
Observe that $\ci y\in \rho(H_\theta)\setminus \R$, due to Proposition~\ref{prp_spec}. We compute, using the notation from \eqref{eq_W}
\begin{align*}
|\langle \psi, H\phi\rangle| 
&\leq 
|\langle \psi, H_\theta\phi\rangle| + |\langle \psi, W_\theta\phi\rangle|\\
&=
|\langle \psi, H_\theta\phi\rangle| + |\langle \psi, (H_\theta-\ci y)(H_\theta-\ci y)^{-1}W_\theta\phi\rangle|\\
&\leq 
\|H_\theta^*\psi\|\|\phi\| + \|(H_\theta^*+ \ci y)\psi\|\|(H_\theta-\ci y)^{-1}W_\theta\phi\|.
\end{align*}
Note that
\begin{equation*}
H_\theta-\ci y = (H-\ci y)(1 + (H-\ci y)^{-1}W_\theta)
\end{equation*}
and that, recalling \eqref{eq:ChoiceOfRp}, \eqref{eq_W_est} and \eqref{eq:choiceof-y},
\begin{equation*}
\bigl\|(H-\ci y)^{-1}W_\theta\bigr\| \leq \left( \sup_{x\in \R}\frac{x^2+1}{x^2 +y^2}\right)^{1/2} \bigl\|(H-\ci)^{-1}W_\theta\bigr\| \leq \frac12.
\end{equation*}

Abbreviating $B_\theta = (1+(H-\ci y)^{-1}W_\theta)^{-1}$, we may estimate
\begin{equation*}
\|(H_\theta-\ci y)^{-1}W_\theta\phi\| \leq \|B_\theta\| \|(H-\ci y)^{-1}W_\theta\phi\|\leq C \|\phi\|.
\end{equation*} 
Hence, there exists a $C_\psi>0$ such that
\begin{equation*}
\forall \phi\in D(H):\quad \bigl|\langle \psi, H\phi\rangle\bigr| \leq C_\psi\|\phi\|,
\end{equation*}
and therefore we may conclude that $\psi \in D(H^*) = D(H)$, exploiting the self-adjointness of $H$. This shows that $D(H_\theta^*) = D(H)$ and that
 $H_\theta^*= H_{\overline{\theta}}$. 
\end{proof}

\subsection{The Mourre Estimate}\label{Subsect-ME}

At this stage we will single out a specific energy $\lambda_0\in\R$, where we shall assume that $H$ has an eigenvalue.
In order for the dilated Hamiltonian to have its essential spectrum out of the way of the eigenvalue, we shall impose a Mourre estimate locally around $\lambda_0$. To formulate the requirement, we need the notation $E_H(B)$ for the spectral projection associated with a Borel set $B\subset \R$ and 
the self-adjoint operator $H$.

\begin{cond}\label{cond_Mourre_estimateNEW}
Let $\lambda_0\in\R$.
For the pair of self-adjoint operators $H$ and $A$ satisfying Condition~\ref{cond_g}, we further assume:
\begin{enumerate}
\item\label{item:Cond-ev} $\lambda_0\in\sigma_\mathrm{pp}(H)$.
\item\label{item:CondMourre} There exist $e,C,\kappa>0$ and a compact operator $K$, such that
\begin{equation}\label{ME:With-K}
\ci \ad_A(H) \geq e - C E_H(\R\setminus [\lambda_0-\kappa,\lambda_0+\kappa])\langle H\rangle - K
\end{equation}
in the sense of quadratic forms on $D(H)$.
% \item\label{cond_g_conj}  We suppose that there exists a conjugation\footnote{see Definition \ref{def_intertw}.} $\conj$ on $\cH$ satisfying
% $\conj D(H) \subset D(H)$, $\conj D(A)\subset D(A)$,
% \begin{equation*}
% \conj H =  H \conj \quad \textup{and}\quad \conj A = - A \conj.
% \end{equation*}
\end{enumerate}
\end{cond}

\begin{nota} We write $P_0 = E_H(\{\lambda_0\})$ for the orthogonal projection onto the eigenspace of $H$ associated with the eigenvalue $\lambda_0$. Furthermore, we abbreviate $\bP_0 = 1 - P_0$ for the projection onto the orthogonal complement of the eigenspace.
\end{nota}

\begin{rems}
\begin{enumerate}
\item Observe that it is a consequence of Conditions~\ref{cond_g}.\ref{cond_g_C1A},   \ref{cond_Mourre_estimateNEW}.\ref{item:CondMourre} and the Virial Theorem \cite{GeoGer_virial} that $P_0$ is a finite rank projection.
\item Choosing $\kappa$ possibly smaller, one may replace the compact operator $K$ in \eqref{ME:With-K} with 
a positive multiple of the eigenprojection $P_0$. More precisely,
\begin{equation}\label{ME:With-P}
\ci \ad_A(H) \geq e' - C' \bigl(E_H(\R\setminus [\lambda_0-\kappa',\lambda_0+\kappa'])\langle H\rangle + P_0\bigr),
\end{equation}
for suitably chosen constants $e'\in (0,e]$, $\kappa'\in (0,\kappa]$ and $C'\geq C$. 
It is in this form that we shall use the Mourre estimate, and for convenience we assume $\kappa'\leq \sqrt{3}$.
\item Equation \eqref{ME:With-P} differ from the more usual version of Mourre's estimate:
\begin{equation*} %\label{eq_ME_standard}
E_H(I)\ci\ad_A(H)E_H(I) \geq eE_H(I) - K,
\end{equation*}
where $I = [\lambda_0-\kappa, \lambda_0+\kappa]$.
Under Condition~\ref{cond_g}, most notably the consequence that $\ad_A(H)$ is an $H$-bounded operator, the two estimates are equivalent. This would be false if the factor $\langle H\rangle$  is replaced by $1$ on the right-hand side of \eqref{ME:With-P}.
% When passing from one to the other ne may have to alter the %constants.  \eqref{ME:With-K} follows from the latter by %expanding $1=E_H(I) + 1- E_H(I)$ on both sides of $\ci\ad_A(H)$. %The term $E_H(I)\ci\ad_A(H)E_H(I)$ is estimated with the aid of %\eqref{eq_ME_standard}. The terms $E_H\ci\ad_A(H)(1-E_H(I))$ can %be estimated from below by $-C\langle H\rangle (1-E_H(I))$. For %the last term we note that
%\begin{align*}
%(1-E_H(I))\ci\ad_A(H)(1-E_H(I)) = (1-E_H(I))\langle %H\rangle^{1/2} T \langle H\rangle^{-1/2} (1-E_H(I)),
%\end{align*}
%where $T=\langle H\rangle^{-1/2} \ci\ad_A(H) \langle %H\rangle^{-1/2}$. Since both $\ci\ad_A(H)\langle H\rangle^{-1}$ %and $\langle H\rangle^{-1}\ci\ad_A(H)$ are bounded operators, $T$ %can be seen to be bounded by using interpolation theory between %rigged Hilbert spaces. Thus, $(1-E_H(I))\ci\ad_A(H)(1-E_H(I))$ %can be estimated from below by $-C'\langle H\rangle$ as well. %Putting everything together we obtain \eqref{ME:With-K}.
\end{enumerate} 
\end{rems}

As a preparation for a Feshbach analysis, we have:

\begin{lem}\label{lem_bPHbP} Assume Conditions \ref{cond_g},~\ref{cond_Mourre_estimateNEW}.\ref{item:Cond-ev}
and~\ref{cond_Mourre_estimateNEW}.\ref{item:CondMourre}.
The following three statements are true for all $\theta\in B_{R'}^\C(0)$:
\begin{enumerate}
\item\label{item:ReducedClosed} $\bP_0 H_\theta \bP_0$ is a closed operator with dense domain $\bP_0 D(H)$.
\item\label{item:ReducedAdjoint} $[\bP_0 H_\theta\bP_0]^*= \bP_0 H_{\bar{\theta}}\bP_0$ on $\bP_0 D(H)$.
\item\label{item:ReducedSpectrum} For all $\theta\in B_{R'}^\C(0)$: 
$\sigma(\bP_0 H_\theta\bP_0) \subset \bigl\{ x+\ci y  \, \big| \, |y| \leq 4C|\theta|(|x|+1) \bigr\}$.
\end{enumerate}
\end{lem}

\begin{proof} As for \ref{item:ReducedClosed}, note first that $H_\theta \bP_0$ with domain $D(H)$ is closed, since $\bP_0 D(H) \subset D(H)$ and  $H_\theta$ with domain $D(H)$ is a closed operator (Proposition~\ref{prp_equiv}). To conclude, observe that
the graph of $\bP_0 H_\theta \bP_0$ is the range of the open map $\cH\oplus \cH\ni (\psi,\varphi)\to 
(\bP_0\psi,\bP_0\varphi)\in \bP_0\cH \oplus \bP_0\cH$ applied to the graph
of $H_\theta\bP_0$.

We turn to the claim \ref{item:ReducedAdjoint}. Clearly, $\bP_0 H_{\bar{\theta}} \bP_0 \subset [\bP_0 H_\theta\bP_0]^*$. Let $\varphi \in D([\bP_0 H_\theta\bP_0]^*)$  viewed as an element of $\bP_0\cH \subset \cH$, and compute for $\psi\in D(H)$:
\[
\begin{aligned}
\langle \varphi, H_\theta \psi\rangle & = \langle \bP_0 \varphi, H_\theta (\bP_0 + P_0)\psi\rangle \\
& =  \langle \varphi, \bP_0 H_\theta \bP_0\psi\rangle +  \langle \bP_0 \varphi, H_\theta P_0\psi\rangle.
\end{aligned}   
\]
Since $P_0$ is finite rank operator and $H_\theta$ is closed, it follows from the Closed Graph Theorem that 
$H_\theta P_0$ is bounded. Hence, there exists $C>0$ such that
\[
\bigl| \langle \varphi, H_\theta \psi\rangle \bigr| \leq C \|\psi\|,
\]
which implies that $\varphi \in D((H_\theta)^*) = D(H_{\bar{\theta}}) = D(H)$. Here we used 
Lemma~\ref{lem_adjoint}. Since $\bP_0\cH\cap D(H) = \bP_0 D(H)$, we are done.

The last claim \ref{item:ReducedSpectrum} may be established by repeating the proof of Proposition~\ref{prp_spec}.
\end{proof}

In formulating the following proposition and in its proof, we make use of the eigenvalue $\lambda_0$  from Condition~\ref{cond_Mourre_estimateNEW}
and the constants $e'$ and $\kappa'$ from \eqref{ME:With-P}. The radius $R'$ 
was defined in \eqref{eq:ChoiceOfRp}. For the open upper half-plane, we use the notation
\begin{equation}
\label{Cplus}
\C^+ := \bigl\{z\in \C \, \big| \, \Im(z)>0\bigr\}.
\end{equation}
  
\begin{prp}\label{prp_windowNEW} Assume Conditions~\ref{cond_g} and~\ref{cond_Mourre_estimateNEW}.  
Abbreviate for $\sigma,\rho >0$ and $\theta \in \C^+$:
  \begin{equation}\label{eq:Rectangle}
    \cR_\theta(\sigma,\rho)=\bigl\{z\in\C\,\big|\,\Re(z)\in (\lambda_0-\rho,\lambda_0+\rho),\Im(z)\in(-\sigma\Im(\theta),\infty)\bigr\},
  \end{equation}
  There exist constants $R'',\rho>0$ with $R'' \leq R'$, such that
  \begin{equation}\label{eq_window}
    \forall\theta\in B_{R''}^\C(0)\cap \C^+\colon \quad 
    \cR_\theta(e'/2,\rho)\cap\sigma(\bP_0 H_\theta \bP_0)=\emptyset.
  \end{equation}
\end{prp}

\begin{proof} Using the constants from \eqref{ME:With-P}, we define a bounded operator
\begin{equation}\label{eq_windowproof_L-NEW}
L:= C' E_H(\R\setminus [\lambda_0- \kappa',\lambda_0+\kappa'])\langle H \rangle ( H-\lambda_0)^{-1}.
\end{equation}
Note that $\|L\| \leq 4 C' \langle \lambda_0\rangle/\kappa'$, where we used $\kappa'\leq \sqrt{3}$
and that $\langle \lambda\rangle \leq 2 \langle\lambda_0\rangle\langle \lambda-\lambda_0\rangle$. We claim suitable choices
\begin{equation}
 \begin{aligned}\label{rho-and-Rpp}
\rho  &= \min\Bigl\{1,\frac{e'}{16 C' \langle \lambda_0\rangle/\kappa'}\Bigr\},\\
R''  &= \min\Bigl\{R', \frac{e'}{24 C(5|\lambda_0|+ 11)( C' \langle \lambda_0\rangle/\kappa' + C)}\Bigr\},
 \end{aligned}
 \end{equation}
where $C$ and $ R'$ were defined in \eqref{eq:ChoiceOfC} and \eqref{eq:ChoiceOfRp}, respectively. Recall that
$R'' C\leq R' C\leq  1/3$. 

Let $\theta\in\C^+\cap B_{R''}^\C(0)$ and $\mu\in \cR_\theta(e'/2,\rho)\cap\sigma(\bP_0H_\theta \bP_0)$.
Note that due to Lemma~\ref{lem_bPHbP}.\ref{item:ReducedSpectrum}, we may estimate
\begin{equation}\label{eq_mu_upper_bound}
|\mu| \leq \bigl(|\lambda_0|+\rho + 1\bigr)\bigl(1+ 16 R'^2 C^2\bigr)^{1/2}\leq 2|\lambda_0|+ 4.
\end{equation}

By 
%Remark~\ref{rem:barH}.\ref{item:barHandC} and 
Lemma~\ref{lem_bPHbP}.\ref{item:ReducedClosed}, we may apply Lemma~\ref{lem-dich} to the operator $T=\bP_0 H_\theta\bP_0$ acting in $\bP_0\cH$. Assume first that there exists a sequence $\psi_n\in \bP_0 D( H)$ with $\|\psi_n\|=1$, such that 
\begin{equation}\label{Approx-ev}
o_n := \bigl\|\bP_0(H_\theta-\mu) \bP_0\psi_n\bigr\|\to 0 \quad \textup{for } n\to\infty.
\end{equation} 

We estimate for all $n$ using \eqref{RelBounds} and \eqref{eq_W_est} (recalling that $C|\theta|\leq C R''\leq 1/3$)
\begin{align}\label{UnifBndOnSeq}
\nonumber \bigl\|\bP_0\psi_n\bigr\|_H & \leq 2 \bigl\|\bP_0\psi_n\bigr\|_{H_\theta}\\
\nonumber & \leq 2\bigl(\bigl\|\bP_0H_\theta\bP_0\psi_n\bigr\| + \bigl\|P_0 W_\theta \bP_0\psi_n\bigr\|+1\bigr)\\
\nonumber &\leq 2\bigl(o_n + |\mu| + \frac12 (|\lambda_0|+1)+1\bigr)\\
& = 2o_n + 5|\lambda_0| + 11.
\end{align}
Here we used \eqref{eq_mu_upper_bound} in the last step.
Exploiting the power series expansion \eqref{PowerSeries} of $H_\theta$, the Mourre estimate \eqref{ME:With-P} and simplifying for real expectation values, we obtain for any $n$
\begin{align} \label{eq_windowproof1}
\Im(\mu) & =\Im\inner[\big]{\bP_0\psi_n}{(\mu-H_\theta)\bP_0\psi_n}+\Im\inner[\big]{\bP_0\psi_n}{H_\theta\bP_0\psi_n}\nonumber\\
    & =\Im\inner[\big]{\bP_0\psi_n}{(\mu-H_\theta)\bP_0\psi_n} -\Im\inner[\big]{\bP_0\psi_n}{ \theta\ci\ad_A(H)\bP_0\psi_n}\nonumber\\
    &\quad -\Im\inner[\Big]{\bP_0\psi_n}{\sum_{k=2}^\infty\frac{(-\theta)^k}{k!}\ci^k\ad_A^k(H)\bP_0\psi_n}\nonumber\\
    & =\Im\inner[\big]{\bP_0\psi_n}{(\mu-H_\theta)\bP_0\psi_n} - \Im(\theta)\inner[\big]{\bP_0\psi_n}{\ci\ad_A(H)\bP_0\psi_n}\nonumber\\
    &\quad -\sum_{k=2}^\infty\frac{\Im\bigl((-\theta)^k\bigr)}{k!}\inner[\Big]{\bP_0\psi_n}{\ci^k\ad_A^k(H)\bP_0\psi_n}\nonumber\\
& \leq o_n
 - \Im(\theta)\left[e' -C'\inner[\big]{\bP_0\psi_n}{E(|H-\lambda|\geq \kappa')\langle H\rangle\bP_0\psi_n}\right] \nonumber\\
    &\quad -\sum_{k=2}^\infty\frac{\Im\bigl( (-\theta)^k\bigr)}{k!}\inner[\Big]{\bP_0\psi_n}{\ci^k\ad_A^k(H)\bP_0\psi_n}. 
\end{align}
Note that for all $k$, we have
$|\Im((-\theta)^k)|
\leq  2^k |\Im(\theta)| |\theta|^{k-1}$.
Therefore,
\begin{align}\label{eq_windowproof2}
&\Bigl|\sum_{k=2}^\infty\frac{\Im\bigl((-\theta)^k\bigr)}{k!}\inner[\Big]{\bP_0\psi_n}{\ci^k\ad_A^k(H)\bP_0\psi_n}\Bigr|
\leq
\sum_{k=2}^\infty \bigl|\Im\bigl((-\theta)^k\bigr)\bigr|C^k\bigl\|\bP_0\psi_n\bigr\|_{H}\nonumber\\
& \quad \leq 
C \bigl|\Im(\theta)\bigr| \sum_{k=2}^\infty 2^{k-1}\bigl|\theta\bigr|^{k-1}C^{k-1}\bigl\|\bP_0\psi_n\bigr\|_{H}\nonumber\\
\nonumber & \quad =
C \bigl|\Im(\theta)\bigr|\frac{2C|\theta|}{1-2C|\theta|}\bigl\|\bP_0\psi_n\bigr\|_{H}\\
 & \quad \leq 6 \bigl|\Im(\theta)\bigr| R'' C^2 \bigl(2o_n + 5|\lambda_0| +11\bigr),
\end{align}
where we used \eqref{UnifBndOnSeq} and that $C|\theta| \leq 1/3$ in the last step. 
%Define the bounded operator\fxnote{remove this version of $L$}
%\begin{equation}
%L:= C_0' E(|H(\xi_0)-\lambda(\xi_0)|\geq \kappa_0)\langle H(\xi_0) \rangle ( H(\xi_0)-\lambda(\xi_0))\inv.
%\label{eq_windowproof_L}
%\end{equation}
We estimate using \eqref{eq_W_est},
recalling the definition  \eqref{eq_windowproof_L-NEW} of the bounded self-adjoint operator $L$,
\begin{align}\label{eq_windowproof3}
\nonumber &
C'\bigl|\bigl\langle\bP_0\psi_n, E_H(\R\setminus [\lambda_0-\kappa',\lambda_0+\kappa'])\langle H\rangle\bP_0\psi_n\bigr\rangle\bigr| \\
\nonumber&\quad = \bigl| \bigl\langle\bP_0 \psi_n, L \bP_0 (H-\lambda_0)\bP_0\psi_n\bigr\rangle\bigr| \\
&\nonumber\quad \leq  \bigl| \bigl\langle\bP_0 \psi_n, L \bP_0 (H-\mu)\bP_0\psi_n\bigr\rangle\bigr|+
|\lambda_0-\re(\mu)| \bigl| \bigl\langle \bP_0\psi_n, L\bP_0\psi_n\bigr\rangle\bigr| \\
\nonumber&\quad \leq
\bigl\|L\bigr\|\bigl\|\bP_0(H_\theta-\mu)\bP_0\psi_n\bigr\| + \bigl\|L\bigr\| \bigl\| W_\theta\bP_0\psi_n\bigr\| + |\lambda_0-\re(\mu)|\bigl\|L\bigr\|\\
\nonumber&\quad \leq
\bigl\|L\bigr\| o_n + \bigl\|L\bigr\| \bigl\| W_\theta (H+\ci)\inv\bigr\|\bigl\|\bP_0\psi_n\bigr\|_{H} + \rho\bigl\|L\bigr\|\\
\nonumber&\quad \leq
\|L\| o_n + \frac32 C R''  \|L\| \|\bP_0\psi_n\|_{H}  + \rho\|L\|\\
&\quad  \leq \bigl\|L\bigr\|\bigl(1+ 3 C R'' \bigr)o_n + \frac32\bigl(5|\lambda_0|+11\bigr)CR''\bigl\|L\bigr\| + \rho\bigl\|L\bigr\|,
\end{align}
where we used \eqref{UnifBndOnSeq} in the final step.

Combining (\ref{eq_windowproof1}), (\ref{eq_windowproof2}) and (\ref{eq_windowproof3}) we obtain
\begin{align*}
\Im(\mu) &\leq  - \Im(\theta) \bigl(e' -  \frac32  C R'' (5|\lambda_0|+11)(\|L\|+ 4  C) - \rho \|L\| \bigr)\\
&\quad  +  
\bigl(1+  |\Im(\theta)|(12 R'' C^2  +\|L\|(1+ 3 C R'' ))  \bigr)o_n.
\end{align*}
By the choices of $\rho$ and $R''$ from \eqref{rho-and-Rpp} and the estimate $\|L\| \leq 4 C' \langle \lambda_0\rangle/\kappa'$, we observe that
\begin{equation*}
\frac32 R''C (5|\lambda_0|+11)(\|L\|+ 4  C) + \rho \|L\| \leq \frac{e'}{2}
\end{equation*}
and thus, taking the limit $n\to \infty$ using \eqref{Approx-ev}, 
we arrive at
\begin{equation}\label{ContraBound1}
\Im(\mu) \leq -\Im(\theta)\frac{e'}{2}.
\end{equation}
This estimate contradicts the choice of $\mu\in \cR_\theta(e'/2,\rho)$, cf.~\eqref{eq:Rectangle}.

If \eqref{Approx-ev} does not hold, then by Lemma~\ref{lem_bPHbP}.\ref{item:ReducedAdjoint} and Lemma~\ref{lem-dich}, there exists a sequence $\phi_n\in \bP_0 D(H_{\bar{\theta}})$, with $\|\phi_n\| = 1$ all $n$ and 
\begin{equation}
o_n := \|\bP_0 (H_{\bar{\theta}}-\bar{\mu})\bP_0)\phi_n\| \to 0, \quad \textup{for }n\to\infty.
\end{equation}
We now repeat the estimates \eqref{UnifBndOnSeq}--\eqref{ContraBound1}, replacing $\mu$ by $\bar{\mu}$ 
and $\theta$ by $\bar{\theta}$, and recalling that $\Im(\bar{\theta})>0$. This results in the estimate
\[
\Im(\bar{\mu}) \geq -\Im(\bar{\theta}) \frac{e'}2.
\]
Hence $\Im(\mu) \leq -\Im(\theta) e'/2$, which  completes the proof.
\end{proof}

In the following, we use the definition $\sigma_{\mathrm{ess}}(H) = \sigma(H)\setminus\sigma_{\mathrm{disc}}(H)$, where $\sigma_{\mathrm{disc}}(H)$ is the set of all isolated points $\lambda\in\sigma(H)$ such that when $\Gamma$ is a counterclockwise loop around $\lambda$, which separates $\lambda$ from the rest of the spectrum, the Riesz projection
\begin{equation*}
\frac1{2\pi \ci}\int_\Gamma(z-H)\inv \mathrm{d}z
\end{equation*} 
onto the generalized eigenspace has finite rank.

The following theorem is proven using the Feshbach reduction method, for which Proposition~\ref{prp_windowNEW} above is an essential prerequisite. 

\begin{thm}\label{thm_ess_spec}
Assume Conditions \ref{cond_g} and \ref{cond_Mourre_estimateNEW}. Then 
\begin{equation*}
\forall \theta \in B_{R''}^{\C}(0)\cap \mdsc^+:\quad \sigma_{\mathrm{ess}}(H_\theta)\cap \cR_\theta(e'/2,\rho) =\emptyset.
\end{equation*}
The constants $\rho,R''$ and the sets $\cR_\theta$ come from Proposition~\ref{prp_windowNEW}.
\end{thm}

\begin{proof}  By Proposition \ref{prp_windowNEW} there exist $R'',\rho>0$  such that for all $|\theta|<R''$ the closed operator $\bP_0 H_\theta\bP-z\bP_0$ is invertible on $\bP_0\mathcal{H}$ for all $z\in \cR:=\cR_\theta(e'/2,\rho)$. Define reduced resolvents
\begin{equation*}
\overline{R}_\theta(z):=\bigl(\bP_0 H_\theta\bP_0 -z\bP_0\bigr)^{-1}
\end{equation*}
on $\bP_0\cH$ for $z\in\cR$. Recall that $W_\theta$ is defined in (\ref{eq_W}). For $z\in \cR$ we can construct the Feshbach map on the finite dimensional subspace $P_0\cH$:
\begin{align*}
F_{P_0}(z)
&=
P_0(H_\theta-z)P_0 - P_0H_\theta\bP_0\overline{R}_\theta(z)\bP_0 H_\theta P_0\\
&=
P_0(W_\theta+\lambda_0-z)P_0 - P_0W_\theta\bP_0\overline{R}_\theta(z)\bP_0 W_\theta P_0.  
\end{align*}
Clearly, $F_{P_0}(z)$ is a finite rank operator, which can be interpreted as a matrix, and hence;
by isospectrality of the Feshbach reduction \cite{DJ,GrieHas_smoothFeshbach},
\begin{equation*}
\mu\in \sigma(H_\theta)\cap\cR  \Leftrightarrow \det\bigl(F_{P_0}(\mu)\bigr) =0.
\end{equation*}
Since for $|\re\mu-\lambda_0|< e'/2$ and $\im\mu$ large it holds that $\mu\not\in\sigma(H_\theta)$ (cf.~Proposition~\ref{prp_spec}), we conclude
by the Unique Continuation Theorem for holomorphic functions that the set  $\sigma(H_\theta)\cap \mathcal{R}$ is locally finite.
Note that $\mu\in\sigma(H_\theta)\cap\cR$ is necessarily an eigenvalue for $H_\theta$. 
In order to establish the theorem, it remains to prove that the Riesz projections pertaining to the
eigenvalues in $\cR$ are of finite rank. Let $\mu\in \cR\cap \sigma(H_\theta)$ and choose $r>0$, such that
$D\subset \cR\setminus\sigma(H_\theta)$, where $D = \{ z\in\C \,|\, 0 <|z-\mu|\leq r\}$ denotes a closed punctured disc.  

 The inverse of $F_{P_0}(z)$ for $z\in D$ has a Laurent expansion
\begin{equation*}
F_{P_0}(z)^{-1} = \sum_{k=1}^{N}B_{-k}(z-\mu)^{-k} +\sum_{k=0}^{\infty}B_k(z-\mu)^k
\end{equation*}
convergent in the punctured disc $D$. Here $N\geq 1$ and $\{B_k\}_{k=-N}^\infty$ denote linear operators on $P_0\cH$.
See \cite[Sect.~6.1]{RaSt}. Note that the inverse has no essential singularities since we are in finite dimension.

By \cite{DJ,GrieHas_smoothFeshbach}, for $z\in\cR\setminus \sigma(H_\theta)$, the inverse $R_\theta(z)$ of $H_\theta-z$ can be recovered from the inverse Feshbach operator 
and the reduced resolvent via the block decomposition
\begin{align*}
P_0 R_\theta(z)P_0 &=F_{P_0}(z)^{-1},\\
P_0R_\theta(z)\bP_0 &=-F_{P_0}(z)^{-1}P_0 W_\theta\bP_0\overline{R}_\theta(z),\\
\bP_0 R_\theta(z)P_0 &= -\overline{R}_\theta(z) \bP_0 W_\theta P_0 F_{P_0}(z)^{-1}, \\
\bP_0 R_\theta(z)\bP_0 &= \overline{R}_\theta(z) + \overline{R}_\theta(z)\bP_0 W_\theta P_0 F_{P_0}(z)^{-1} P_0 W_\theta\bP_0 \overline{R}_\theta(z).
\end{align*}
Note that the map $z\mapsto \overline{R}_\theta(z)$ is analytic in $\cR$, so the only singularities are those in $\sigma(H_\theta)$, coming from the inverse Feshbach operator.

Let $\gamma\colon [0,2\pi]\to \C$ be the closed curve $\gamma(t) = \mu + r e^{\ci t}$ parametrizing the (outer) boundary of $D$, encircling $\mu$. Recall the
construction of the Riesz projection
\begin{equation*}
P_\theta(\mu) = -\frac{1}{2\pi \ci}\int_{\gamma} R_\theta(z)\,\mathrm{d}z
\end{equation*}
associated with the eigenvalue $\mu$. The block decomposition of $R_\theta(z)$ induces a block decomposition of  $P_\theta(\mu)$ and the Riesz projection has finite rank, provided $\bP_0 P_\theta(\mu) \bP_0$ is of finite rank. To check this, we compute
\begin{align*}
&- \bP_0 P_\theta(\mu) \bP_0 \\
&\quad = \frac1{2\pi \ci}
\int_{\gamma} \left[\overline{R}_\theta(z) +\overline{R}_\theta(z)\bP_0 W_\theta P_0 F_{P_0}(z)^{-1} P_0 W_\theta \bP_0 \overline{R}_\theta(z) \right]\, \mathrm{d}z\\
&\quad =
\sum_{k=1}^{N} \frac1{2\pi \ci}\int_{\gamma}(z-\mu)^{-k}\overline{R}_\theta(z)\bP_0 W_\theta P_0 B_{-k} P_0 W_\theta \bP_0 \overline{R}_\theta(z)\, \mathrm{d}z\\
&\qquad+
 \frac1{2\pi \ci}\int_{\gamma}\sum_{k=0}^{\infty}(z-\mu)^k\overline{R}_\theta(z)\bP_0 W_\theta P_0 B_k P_0 W_\theta \bP_0 \overline{R}_\theta(z)\, \mathrm{d}z,
\end{align*}
where we have used that the function $\cR\ni z\mapsto \overline{R}_\theta(z)$ is analytic. Moreover, the integral in the last line of the equation above is carried out over an analytic function, once again, and thus equals $0$. The remaining $N$ singular integrals can be evaluated by Cauchy's Integral Formula:
\begin{align*}
&\frac1{2\pi \ci}\int_{\gamma}(z-\mu)^{-k}\overline{R}_\theta(z)\bP_0 W_\theta P_0 B_{-k} P_0 W_\theta \bP_0 \overline{R}_\theta(z) \,\mathrm{d}z\\
&\quad =
\frac1{(k-1)!}\frac{\mathrm{d}^{k-1}}{\mathrm{d}z^{k-1}} \overline{R}_\theta(z)\bP_0 W_\theta P_0 B_{-k} P_0 W_\theta \bP_0 \overline{R}_\theta(z)\bigg|_{z=\mu}\\
&\quad =\frac1{(k-1)!}
\sum_{j=0}^{k-1}\binom{k-1}{j}(-1)^{k-1} j!(k-1-j)!\\
& \qquad \qquad \times \overline{R}_\theta(\mu)^{1+j}\bP_0 W_\theta P_0 B_{-k} P_0 W_\theta \bP_0 \overline{R}_\theta(\mu)^{k-j}\\
&\quad =
(-1)^{k-1}  \sum_{j=0}^{k-1} \overline{R}_\theta(\mu)^{j+1}\bP_0 W_\theta P_0 B_{-k} P_0 W_\theta \bP_0 \overline{R}_\theta(\mu)^{k-j}.
\end{align*}
Since each term in the sum above is a finite rank operator, we conclude that $\bP_0 P_\theta(\mu) \bP_0$ is of finite rank.

Since $\sigma(H)\cap \cR$ is locally finite  and all the associated Riesz projections have
finite rank, we have shown that $\sigma_{\mathrm{ess}}(H_\theta)\cap \cR = \emptyset$. This completes the proof. 
\end{proof}

Note that 
\[
D(U(\theta)) = \Bigl\{\psi\in\cH \,\Big|\, \int_{\R} e^{2\Im(\theta)x} \,\mathrm{d}E_\psi(x)<\infty\Bigr\},
\]
where $E_\psi$ is the spectral measure for $A$ associated with the state $\psi$. Motivated by this we abbreviate for $r\geq 0$:
\[
D_r(A) =  \Bigr\{\psi\in\cH \,\Big|\, \int_{\R} e^{2 r |x|} \,\mathrm{d}E_\psi(x)<\infty\Bigr\}.
\]
 Having established Theorem \ref{thm_ess_spec}, we may conclude the following theorem by invoking a general result
of Hunziker and Sigal \cite[Theorem~5.2]{HS}.

\begin{thm}\label{thm_HuSi_appl}
Assume Conditions \ref{cond_g} and \ref{cond_Mourre_estimateNEW}. Let $\theta\in B_{R''}^\C(0)\cap \C^+$.
Then the dilated Hamiltonian $H_\theta$ has an isolated eigenvalue at $\lambda_0$. Denote by  $P_\theta$ the associated  Riesz projection. The following statements hold true:
\begin{enumerate}
\item $\Ran(P_\theta)$ is the eigenspace of $H_\theta$ pertaining to the eigenvalue $\lambda_0$.
\item $P_0 = U(-\theta) P_\theta U(\theta)$ as a form identity on $D_{|\Im(\theta)|}(A)$.
\item $\Rank(P_0) = \Rank(P_\theta)$.
\item Let $r<R''$. Then $\Ran(P_0) \subset D_r(A)$.
\end{enumerate}
\end{thm}

\begin{rem} The above theorem implies that eigenfunctions pertaining to the eigenvalue $\lambda_0$ 
are analytic vectors for the operator $A$. A result previously established by brute force in \cite{MoWe} under a slightly weaker condition. % Here, however, we needed  Condition~\ref{cond_Mourre_estimateNEW}.\ref{cond_g_conj}, which played no role for the method employed in \cite{MoWe}.
\end{rem}

\subsection{Analytic Perturbation Theory}\label{Subsec-AnPert}

\begin{cond}\label{cond_analytic} Let $\lambda_0\in\R$,
  $\xi_0\in\R^d$, and $U\subset\mdsr^d$ an open (connected)
  neighborhood of $\xi_0$, $A$ a self-adjoint operator on $\cH$ and
  $\{H(\xi)\}_{\xi\in U}$ a family of self-adjoint operators on $\cH$.
\begin{enumerate}
\item\label{item-AnaCond-1} $D(H(\xi)) = D(H(\xi_0)) =: \dom$ for all $\xi\in U$.
\item\label{item-AnaCond-2} For all $\xi$ in $U$, the operator $H(\xi)$ satisfies Condition \ref{cond_g} with the same constants $R$ and $M$.
\item\label{item-AnaCond-3} The triple $\lambda_0$, $A$ and $H(\xi_0)$ satisfies Condition \ref{cond_Mourre_estimateNEW}.
\item\label{cond_g_fam} There exists $\theta_0\in B_R^\C(0)$ with $\Im(\theta_0) > 0$, such that the map
$\xi\to H_{\theta_0}(\xi)$ extends from $U$  to an analytic family of Type (A) defined for $\xi\in U_\C\subset \C^d$,  an open (connected) set with $U\subset U_\C\cap\R$.
\end{enumerate}
\end{cond}

\begin{rem} Suppose one strengthens 
Condition~\ref{cond_analytic} and assumes that $\xi\to H_\theta(\xi)$ extends to an analytic family of Type (A)
not just for one $\theta_0$ but for all $\theta$ in a complex disc of radius $\Theta< R'$ around $0$. 
Then one may use Morera's theorem to conclude that for any $\psi\in\dom$ and $n$, we have
\[
\ci^n\ad_A^n(H(\xi))\psi = \frac{(-1)^n}{2\pi \ci} \int_{|\theta|=\Theta/2} n! \theta ^{-n-1} H_\theta(\xi)\psi \,\mathrm{d}\theta,
\] 
a priori for real $\xi$, but since the right-hand side extends analytically to  $\xi$ in a complex neighborhood of $\xi_0$,
so does the left-hand side. This will in particular permit one to conclude that also for complex $\xi$ does the
closed operator  $H(\xi)$ iteratively admit commutators with $A$ of arbitrary order. (Note that $H(\bar{\xi})\subset H(\xi)^*$, by unique continuation.) Furthermore, the iterated commutators
must coincide (strongly) with the analytic extension from real $\xi$ of $\ad_A^n(H(\xi))\psi$, obtained above.
\end{rem}

Recall the notation $\lambda_0$ for the eigenvalue of $H(\xi_0)$ with eigenprojection $P_0$. By Theorem~\ref{thm_HuSi_appl},
we know that $\lambda_0$ is an isolated eigenvalue of $H_{\theta_0}(\xi_0)$ with
finite rank eigenprojection $P_{\theta_0}$.  Denote by $n_0$ the common rank of $P_0$ and $P_{\theta_0}$.

Fix $0<\rho'<\rho$, such that 
\begin{equation}\label{rhoprime}
\sigma\bigl(H_{\theta_0}(\xi_0)\bigr)\cap B_{2\rho'}^\C(\lambda_0) = \{\lambda_0\}.
\end{equation}

\begin{rem}\label{rem:HS-lu} We may choose $r'>0$, such that
for all $\xi\in B_{r'}^{\R^d}(\xi_0)$, we have 
\[
\sigma(H_{\theta_0}(\xi))\cap B^\C_{\rho'}(\lambda_0) = \sigma_{\mathrm{pp}}(H_{\theta_0}(\xi))\cap B^\C_{\rho'}(\lambda_0)
\]
 and the total algebraic multiplicity of the eigenvalues in $B^\C_{\rho'}(\lambda_0)$ equals $n_0$ (\!\!\cite[Sect.~IV.4]{K}).
 
By \cite[Theorem~5.2]{HS}, we may now conclude -- just as we did with Theorem~\ref{thm_HuSi_appl} -- that for all $\xi\in B_{r'}^{\R^d}(\xi_0)$:
\begin{equation}\label{eq:RotBack}
\sigma_{\mathrm{pp}}(H(\xi))\cap (\lambda_0-\rho',\lambda_0+\rho') = \sigma(H_{\theta_0}(\xi))\cap  (\lambda_0-\rho',\lambda_0+\rho').
\end{equation}
\end{rem}

If the perturbation parameter $\xi$ is one-dimensional, we may in light of Theorem~\ref{thm_ess_spec} and Condition~\ref{cond_analytic}, invoke 
Kato, in the form of \cite[Theorem~VII.1.8]{K}, and conclude the following theorem.

\begin{thm}\label{thm_kato} Suppose Condition~\ref{cond_analytic} and that $d=1$. There exist 
\begin{itemize}
\item $r>0$ with $(\xi_0-r,\xi_0+r)\subset U$,
\item natural numbers $0\leq m_\pm\leq n_0$ and $n_1^\pm,\dotsc,n_{m_\pm}^\pm\geq 1$ 
with $n_1^\pm+\cdots +n_{m_\pm}^\pm \leq n_0$,
\item   real analytic functions
$\lambda_1^\pm,\dotsc,\lambda_{m_\pm}^\pm\colon I_\pm \to \R$, where $I_- = (\xi_0-r,\xi_0)$ and $I_+=(\xi_0,\xi_0+r)$, satisfying $\lambda_i^\pm(\xi)\neq \lambda_j^\pm(\xi)$ for all $1\leq i < j\leq m_\pm$ and $\xi\in I_\pm$,
\end{itemize}
 such that (recalling $\rho'$ from \eqref{rhoprime})
\begin{enumerate}
\item for any $\xi\in I_\pm$, we have
$\sigma_\mathrm{pp}(H(\xi))\cap (\lambda_0-\rho',\lambda_0+\rho') = \{\lambda_1^\pm(\xi),\dotsc,\lambda^\pm_{m_\pm}(\xi)\}$,
\item $\forall j=1,\dotsc,m_\pm$, we have $\lim_{I_\pm\ni \xi\to\xi_0} \lambda_j^\pm(\xi) = \lambda_0$,
\item the eigenvalue branches
$I_\pm\ni\xi\to \lambda_\pm(\xi)$ have constant algebraic and geometric multiplicity $n_j^\pm$.
\item\label{Ana-d1} Each eigenvalue branch $\lambda^\pm_j\colon I_\pm\to \R$ can be expanded in a convergent Puiseux series near $\xi_0$, that is; a convergent power series expansion in $(\pm(\xi-\xi_0))^{1/\ell}$ for some integer $\ell\geq 1$.
\end{enumerate}
\end{thm}

\begin{rem} If $\lambda_0$ is an isolated eigenvalue, then we know from Kato \cite{K} that the perturbed eigenvalues $\lambda^\pm_i$ are analytic at $\xi=\xi_0$, that is; in \ref{Ana-d1} one may choose $\ell=1$ (and the continuation through $\xi_0$ yields one of the other branches $\lambda^\mp_j$).
We cannot exclude algebraic singularities at $\xi_0$ since the eigenvalues come from $H_{\theta_0}(\xi)$, which may not be a normal operator. We do not know of an example of an embedded eigenvalue where one cannot choose $\ell=1$.
\end{rem}

In the case of multiple parameters, the structure of the point spectrum becomes more complicated, and the right setting here
is that of semi-analytic sets, the definition of which is recalled in Appendix~\ref{sec-semi-sub}. More precisely, we are interested in the analytic structure of the set
\begin{equation*}
\Sigma_\mathrm{pp} := \bigl\{ (\lambda,\xi)\in \R\times U\, \big|\,
\lambda\in \sigma_\mathrm{pp}(H(\xi)) \bigr\}.
\end{equation*}
In the following section, we explore an example where $\xi$ is a total momentum variable, in which case $\Sigma_\mathrm{pp}$ is the energy-momentum point spectrum.
The reader may wish to consult Definitions~\ref{def-semi-sub}.\ref{item-def-OW} and~\ref{def-semi-sub}.\ref{item-def-semi} before proceeding
to the main theorem of this subsection:
%\begin{defi} \begin{enumerate}
%\item Let $W\subset \R^\nu$ be an open set. We write $\mathcal O(W)$ for the smallest ring\footnote{collection of sets %stable under complement as well as under finite intersections and unions.} of subsets of $W$  
%containing sets of the form $\{y\in W \, | \, f(y)>0\}$ and  $\{y\in W \, | \, f(y)=0\}$ where $f$ ranges over real %analytic functions $f\colon W\to \R$.
%\item Let $M\subset \R^\nu$ be an open set.
%A subset $\Sigma\subset M$ is called a semi-analytic subset of $M$
%if: for any $x\in M$, there exists an open neighborhood $W\subset M$ of $x$, such that $\Sigma\cap W\in \mathcal O(W)$.
%\end{enumerate}
%\end{defi} 

\begin{thm}\label{thm_semianalytic} Suppose Condition~\ref{cond_analytic}. There exists $r>0$ and $\rho>0$, such that with $W = (\lambda_0-\rho',\lambda_0+\rho')\times B_r^{\R^d}(\xi_0)$, we have that
$\Sigma_\mathrm{pp}\cap W \in\mathcal O(W)$. In particular, $\Sigma_\mathrm{pp}\cap W$ is a semi-analytic subset of $W$.
\end{thm}

\begin{proof} Let $r'$ be chosen as in Remark~\ref{rem:HS-lu}.
 The projection onto the generalized eigenspace is the Riesz projection,
\[
P(\xi) \equiv P_{\theta_0}(\lambda_0;\xi) = \frac{1}{2\pi \ci}\int_{|z-\lambda_0| = \rho'} (z-H_{\theta_0}(\xi))^{-1}\,\mathrm{d}z,
\]
which depends analytically on $\xi\in B_{r'}^\C(\xi_0)$. Write $V(\xi) = \Ran(P(\xi))$ for the generalized eigenspace 
of dimension $n_0$ and $\Pi\colon \C^{n_0}\to \Pi(\xi_0)$ a linear isomorphism identifying the unperturbed eigenspace with $\C^{n_0}$. Following \cite{GerNier}, we choose $r\in(0,r']$ such that $\|P(\xi)-P(\xi_0)\|\leq 1/2$ for $|\xi-\xi_0|\leq r$. Then $\Theta(\xi) := P(\xi)|_{V(\xi_0)}$ defines 
a linear isomorphism from $V(\xi_0)$ onto $V(\xi)$ with
inverse $\Theta^{-1}(\xi) = (1+ P(\xi_0)(P(\xi)-P(\xi_0)))\inv P(\xi_0)$ and 
\[
\forall \xi\in B_{r}^{\C^d}(\xi_0):\quad T(\xi) = \Pi\inv \Theta(\xi)^{-1} H_{\theta_0}(\xi) \Theta(\xi) \Pi
\]
defines a family of linear operators on $\C^d$ depending analytically on $\xi$ and satisfying that
$\sigma(T(\xi)) = \sigma(H_{\theta_0}(\xi))\cap B^\C_{\rho'}(\lambda_0)$. Hence, recalling \eqref{eq:RotBack},
\[
\Sigma_\mathrm{pp}\cap W = \bigl\{ (\lambda,\xi)\in W \, |\,  \det(T(\xi)-\lambda) = 0\bigr\}.
\]
Here $W$ is defined in the statement of the theorem.
Split into real and imaginary parts  $\det(T(\xi)-\lambda) = u(\lambda,\xi) + \ci v(\lambda,\xi)$, to obtain two real analytic real-valued functions. Then 
\begin{equation*}
\Sigma_\mathrm{pp}\cap W=\bigl\{(\lambda,\xi)\in  W \, \big| \, u(\lambda,\xi)=0\bigr\}\cap\bigl\{(\lambda,\xi)\in W \, \big| \, v(\lambda,\xi)=0  \bigr\}.
\end{equation*}
Since the right-hand side is an element of $\mathcal O(W)$, we are done.
\end{proof}

\begin{rem} In the one dimensional setup, the Puiseux expansion of the eigenvalue branches in particular ensures that 
two distinct branches separate as $c|\xi-\xi_0|^{p/q}$ for some $c\neq 0$ and some natural numbers $p,q$. (This behavior is violated in Example~\ref{ex-Simon})  It is not apparent in the higher dimensional setup that something like this holds true. There is however a remnant in the theory of semi-analytic sets called the Lojasiewicz inequality, which can be interpreted as a statement that two distinct strata that meet at a common boundary stratum, does so with an algebraic lower bound in the distance to the boundary. See \cite[Sect.~7]{BM}.
\end{rem}

\section{Example}\label{sec_ex}

We introduce a two-particle Hamiltonian on $\mathrm{L}^2(\mdsr^d)$ by
\begin{equation*}
H_V '=\omega_1(p_1)+\omega_2(p_2) + V(x_1-x_2),
\end{equation*}
where $p_i =- \mathrm{i}\nabla_{x_i}$, $x_i\in\mathds{R}^d$.
 
 We impose the following set of conditions on $\omega_1,\omega_2$ and $V$:

\begin{cond}[Properties of $\omega_1, \omega_2$ and V]\label{cond_ex}
\leavevmode
\begin{enumerate}
\item\label{item-cond_ex-1} The $\omega_i$'s are real-valued, real analytic functions on $\R^d$ and there exists $\tR>0$, such that the $\omega_i$'s extend to analytic functions in the $d$-dimensional strip 
\begin{equation*}
S_{2\tR}^d:=\bigl\{(z_1,\dots, z_d)\in\C^d\,\big|\,|\im (z_i)| <  2\tR, i=1,\dots, d\bigr\}.
\end{equation*}
We denote the analytic continuations of these functions by the same symbols.
%\label{cond_ex_v_omega}
\item\label{item-cond_ex2} There exist real numbers $s_2\geq s_1 > 0$ and a constant $\tC>0$ such that  
\begin{equation}\label{eq_omega}
|\partial^\alpha\omega_j(k)|\leq \tC \langle k\rangle^{s_j} ,\quad|\omega_j(k)|\geq \frac1{\tC}\langle k\rangle^{s_j} - \tC
\end{equation}
for every multi-index $\alpha\in\N_0^d$, $|\alpha|\leq 1$ and all $k\in S^d_{2\tR}$.
\label{cond_ex_omega_est}\label{cond_ex_bnd_on_v}
%\item 
%The analytic continuation of $v_{\xi}$ is bounded on the $d$-dimensional strip $S_R$, that is 
%\begin{align}
%\sup_{z\in\mathcal{S}_R}|v_{\xi}(z)| \leq M < \infty
%\label{eq_bnd_on_v}
%\end{align}
%for some $M>0$.
%\label{cond_ex_bnd_on_v}
% \item The analytic continuations of the $\omega_j$ still satisfy the bounds (\ref{eq_omega}) on $\mathcal{S}_R$.
%
%\item There exists $q>0$ such that the analytic extensions of $\nabla_k\omega_\xi$ and $v_\xi$ to $\mathcal{S}_R$ satisfy $\sup_{z\in\mathcal{S}_R} |v_\xi(z)\cdot \nabla_k \omega_\xi(z)| \leq q$.
%\label{cond_ex_v_dot_grad}
\item\label{cond_ex_V} Let $d' = 2[d/2]+2$. We suppose that $V\in C^{d'}(\R^d)$ and there exists $a>0$, such that for all $\alpha\in\N_0^d$ with $|\alpha|\leq d'$, we have $\sup_{x\in \R^d}\mathrm{e}^{a|x|}|\partial^\alpha_x V(x)|<\infty$.
%) For all $\alpha\in\mathds{N}_0^d$ with $|\alpha|\leq 1$ we have that $\widehat{x^\alpha V}\in\mathrm{L^2}(\mdsr^d)$.
%\item The Fourier transform $\widehat{V}$ is continuous for $k\neq 0$ and there exist $a > 0$ and $C_V>0$ such that %$\widehat{V}$ extends to an analytic function in $S^d_a\setminus\{0\}$ and the extension satisfies the estimate
%\[
%\forall k\in S^d_a:\qquad  \bigl|\hV(k)\bigr| \leq C_V |k|^{-d}.
%\]
\end{enumerate}
\end{cond}

 Conjugating with the Fourier transform, we see that $H_V'$ is unitarily equivalent to
\begin{equation*}
H_V = \omega_1(k_1) + \omega_2(k_2) + t_V,
\end{equation*}
where $t_V$ is the partial convolution operator
\begin{equation*}
(t_Vf)(k_1,k_2):= \int_{\R^d} \hV(u)f(k_1-u,k_2+u)\, \mathrm{d}u
\end{equation*}
and
\[
\hV(k) = (2\pi)^{-d/2} \int_{\R^d} \mathrm{e}^{-\ci k \cdot x} V(x) \, \mathrm{d} x.
\]
In order to fibrate $H_V$ w.r.t. total momentum $\xi = k_1+k_2$, we introduce a unitary operator $I \colon L^2(\R^d\times\R^d)\to L^2(\R^d;L^2(\R^d))$ by setting
\[
(If)(\xi) = f(\xi-\cdot,\cdot).
\]
%Note that, identifying $\mathrm{L}^2(\mdsr^d\times \mdsr^d) = \int^{\oplus}_{\mdsr^d}\mathrm{L}^2(\mdsr^d)$, the %operator $H_V$ can be fibered w.r.t. $\xi\in\mdsr^d$ by the unitary transform $f(\xi,k)\mapsto f_{\xi}(k)=f(-\xi-k,k)$, %where all fibers are equipped with the Hilbert space $\mathcal{H}_{\xi}=\mathrm{L}^2(\mdsr^d)$. 
Under this transformation, we find that
the Hamiltonian takes the form
\begin{equation*}
I H_V I^* = \int_{\R^d}^{\oplus} H(\xi)\,\mathrm{d} \xi, \quad \textup{where} \quad H(\xi)=\omega_{\xi} + T_{V},
\end{equation*}
and
\begin{equation}\label{Eq-omegaxi}
\omega_{\xi}(k)=\omega_1(\xi-k) + \omega_2(k), \qquad
(T_Vf)(k)=(\widecheck{V}*f)(k).
\end{equation}
Here $\wcV(k)=\hV(-k)$ is the inverse Fourier transform of $V$ and $\wcV *f$ denotes the convolution product. 

We are now in a position to formulate our main result of this section. We introduce the joint
energy-momentum point spectrum
\[
\Sigma_\mathrm{pp} = \bigl\{(\lambda,\xi)\in\R\times\R^d \, \big|\, \lambda\in \Sigma_\mathrm{pp}(\xi)\bigr\},\quad 
\Sigma_\mathrm{pp}(\xi) = \sigma_\mathrm{pp}(H(\xi))
\]
and the energy-momentum threshold set
\[
\begin{aligned}
\mathcal T & = \bigl\{ (\lambda,\xi)\in\R\times\R^d \, \big|\, \lambda\in\mathcal T(\xi)\bigr\},\\
\mathcal T(\xi) & = \bigl\{ \lambda\in\R\,\big|\, \exists k\in\R^d: \ \omega_\xi(k) = \lambda \textup{ and } \nabla_k \omega_\xi(k) = 0\bigr\}.
\end{aligned}
\]
The main result of this section is the following.

\begin{thm}\label{thm-ex} Suppose Condition~\ref{cond_ex}. Then we have
\begin{enumerate}
\item\label{item-thm-ex-1}  $\mathcal T$ is a closed sub-analytic subset of $\R\times\R^d$.
\item\label{item-thm-ex-2} For each $\xi\in \R^d$, the set $\mathcal T(\xi)$ is a locally finite subset of $\R$.
\item\label{item-semian-ex}  $\Sigma_{\mathrm{pp}}\setminus \mathcal T$ is a semi-analytic subset of $(\R\times\R^d)\setminus \mathcal T$.
\item\label{item-thm-ex-4} For each $\xi \in \R^d$, the set  $\Sigma_{\mathrm{pp}}(\xi)\setminus \mathcal T(\xi)$ is a locally finite subset of $\R\setminus \mathcal T(\xi)$. 
\end{enumerate}
\end{thm}

\begin{rem} \begin{enumerate}
\item It is the claim \ref{item-semian-ex}, which is of interest here. The other properties are more or less immediate.
See the proof, which is located at the end of this section. 
\item It remains an open question how bands of non-threshold eigenvalues, as functions of total momentum, may approach the threshold set. For example, under what conditions is $\Sigma_\mathrm{pp}$ a semi- (or sub-) analytic subset of $\R\times \R^d$? Example~\ref{ex-Simon} provides an example where such a result fails to hold true.
\item In \cite{HeSk}, Herbst and Skibsted studies exponential decay of eigenfunctions pertaining to non-threshold eigenvalues of one-body operators of the form $\omega(p) + V$, assuming $\omega$ is a polynomial. Their results apply to $H(\xi)$.  
\end{enumerate}
\end{rem}

We define a self-adjoint operator for every total momentum $\xi\in\R^n$ by
\begin{equation}
A_{\xi}= \frac{\ci}{2}\bigl( v_{\xi}\cdot \nabla_k + \nabla_k\cdot v_{\xi}\bigr),
\label{eq_A}
\end{equation}
where the vector field $v_{\xi}$ is given by
\begin{equation}
v_\xi(k) =  \mathrm{e}^{-k^2-\xi^2}(\nabla_k\omega_\xi)(k).
\label{eq_vectorfield}
\end{equation}
Dilation in momentum space -- of the type considered here -- was previously employed 
for Schr\"odinger operators $-\Delta+V$ in \cite{Nak2}.

In the following, we make frequent use of the estimate
\begin{equation}\label{Peetre}
\forall p\in\R, k,k'\in\R^d:\qquad \langle k+k' \rangle^{p} \leq 2^{|p|}\langle k \rangle^{|p|} \langle k' \rangle^{p},
\end{equation}
referred to as Peetre's inequality in \cite[Lemma~1.18]{X}.

\begin{rem}\label{rem-3.2}
\leavevmode
\begin{enumerate}
\item\label{item-rem3.2.1} By Condition \ref{cond_ex}.\ref{cond_ex_omega_est} $\tC^{-1}\langle k\rangle^{s_j} - \tC\leq |\omega_j(k)| \leq \tC \langle k\rangle^{s_j}$ and thus $D(M_{\omega_j})=D(M_{\langle \cdot \rangle^{s_j}})$. Consequently, cf. \eqref{Peetre}, $D(M_{\omega_\xi})=D(M_{\langle \cdot \rangle^{s_2}})=:\dom$, since $s_2\geq s_1$. Thus all operators $H(\xi)$ have the common domain $\dom$.
\item The two estimates in Condition \ref{cond_ex}.\ref{cond_ex_omega_est} are satisfied by the functions $f(k)=(k_1^2+\cdots + k_d^2)^q$ for $q\in\N$ and $g(k) = (1 + k_1^2+\cdots + k_d^2)^s$ for $s>0$. As for the choice of $\tR$, for the function $f$ any $\tR>0$ will do, whereas for $g$ one must choose $\tR< d^{-1/2}/2$. 
%\item Condition \ref{cond_ex}.\ref{cond_ex_v_omega} implies that $\omega_\xi$ as well as every component of %$\nabla_k\omega_\xi$ extend to analytic functions in the strip $S_R$. The same is true for the vector field $v_\xi$.
\item Condition \ref{cond_ex}.\ref{cond_ex_omega_est} and \eqref{Peetre} %and Condition \ref{cond_ex}.\ref{cond_ex_bnd_on_v}
 imply the existence of $C_\omega,C'_\omega>0$ such that
\begin{equation}\label{eq_bnd_on_v}
\forall \xi,k\in S^d_\tR:\qquad |v_{\xi}(k)| \leq C_\omega < \infty.
\end{equation}
and  
\begin{equation}\label{eq_bnd_on_Dv}
\forall \xi,k\in S^d_\tR:\qquad \|D v_\xi(k)\| \leq C'_\omega   <\infty.
\end{equation}
\item  Condition \ref{cond_ex}.\ref{cond_ex_V} on the potential ensures that $\widehat{V}$ extends analytically to 
the strip $S^d_{a}$. Fix an $a'\in (0,a)$. The assumed decay and smoothness  permits to argue -- using 
integration by parts -- that for some $C_V>0$, we have
\begin{equation}\label{DecayOfV}
\forall k\in S_{a'}^d:\qquad \bigl|\widehat{V}(k)\bigr| \leq C_V \bigl(1+|k|^{d'}\bigr)^{-1}.
\end{equation}
 The choice of $d'$ is made to ensure 
that this estimate implies $\widehat{V}\in L^1(\R^d)$. 
\end{enumerate}
\end{rem}

The next well-known lemma expresses the action of the unitary group generated by $A_\xi$ in  terms of solutions of the ODE generated by the vector field $v_\xi$. See e.g. the PhD thesis of one of the authors \cite[Chap.2, Prop.~2.3]{Ra}.
Since $v_\xi\colon S_\tR^d \to \C^n$, we are led to study the parameter dependent autonomous initial value problem
 \begin{equation}
\frac{\mathrm{d}y}{\mathrm{d}t}  =v_\xi(y),\quad  y(0)=k.
\label{eq_flow_ODE}
\end{equation}
The ODE is defined for $y\in S^d_\tR$. That is, solutions are understood to take values in $S^d_\tR$.

\begin{lem} For $k,\xi\in\R^d$, the initial value problem \eqref{eq_flow_ODE}
admits  a (unique) solution $t\mapsto \gamma^t_\xi(k)$ defined for all time $t\in\R$.
Abbreviating
\begin{equation}\label{eq_flow_J}
J^t_\xi(k)=e^{\int\limits_0^t\nabla\cdot v_\xi(\gamma^s_\xi(k))\mathrm{d}s},
\end{equation}
we have for $f\in L^2(\R^d)$ the formula
\begin{equation}
\left(\mathrm{e}^{\ci t A_\xi}f\right)(k) = \sqrt{J^t_\xi(k)}f(\gamma^t_\xi(k)).
\label{eq_action_dilation}
\end{equation}
\end{lem}

One may -- by direct computation -- verify the useful relation
\begin{equation}
\forall t\in\R,\forall k,\xi\in\R^d:\qquad J^{-t}_\xi(\gamma^{t}_\xi(k))=\frac{1}{J^{t}_\xi(k)}.
\label{eq_J_inv}
\end{equation}
Furthermore, \eqref{eq_bnd_on_v} and the Fundamental Theorem of Calculus ensure that
\begin{equation}\label{FinPropSpeed}
\forall t\in\R,\forall k,\xi\in\R^d:\qquad |\gamma^t_\xi(k)-k| \leq C_\omega |t|.
\end{equation}

\begin{lem}\label{DilatedMultOps} Let $f\in L^\infty_\mathrm{loc}(\R^d)$ and $t\in\R$. Then $f^t:= f\circ\gamma^{t}_\xi\in L^\infty_\mathrm{loc}(\R^d)$,
\[
\mathrm{e}^{\ci t A_\xi} D(M_f) = D(M_{f^t}) \qquad \textup{and}\qquad\mathrm{e}^{\ci t A_\xi} M_f = M_{f^t}\mathrm{e}^{\ci t A_\xi}. 
\]
\end{lem}

\begin{proof} That $f\circ \gamma^{t}_\xi$ is locally essentially bounded follows from \eqref{FinPropSpeed}.
We may also observe that $\mathrm{e}^{\ci t A_\xi} C_0^\infty(\R^d)\subset C_0^\infty(\R^d)$, cf. \eqref{eq_action_dilation} and \eqref{FinPropSpeed}.
Finally, for $\psi\in C_0^\infty(\R^d)$ we have
$M_{f^t} \mathrm{e}^{\ci t A_\xi} \psi = \mathrm{e}^{\ci t A_\xi} M_f \psi$.
The lemma now follows, since $M_{f^t}$ is closed and $C_0^\infty(\R^d)$ is a core for $M_f$. 
\end{proof}

In the following Lemma, we need the strip width $\tR$ from 
Condition~\ref{cond_ex}.\ref{cond_ex_v_omega} and the bound $C_\omega$ on $v_\xi$ from \eqref{eq_bnd_on_v}.

\begin{lem}\label{lem_flow_strip}
Assume Conditions \ref{cond_ex}.\ref{cond_ex_v_omega} and \ref{cond_ex}.\ref{cond_ex_bnd_on_v}. For $k\in\R^d$ and $\xi\in S_{\tilde{R}}^d$,  the solution of (\ref{eq_flow_ODE}), $t\mapsto\gamma^t_\xi(k)$, admits an analytic continuation into the strip $S_r$, where $r = \tR/(C_\omega+1)$.
Moreover, for $0< r' \leq r$, we have $\{\gamma^z_\xi(k) \,|\, z\in S_{r'}\} \subset S^d_{r' C_\omega}\subset S_\tR^d$.
\end{lem}

\begin{proof}[\textsc{Proof.}] By definition, $t\to \gamma^t_\xi(k)$ solves the ODE \eqref{eq_flow_ODE}.

Since $v_\xi$ is analytic in $S^d_\tR \subset S^d_{2\tR}$ as a function of $k$, the function $t\to \gamma^t_\xi(k)$  admits -- by Cauchy-Kowaleskaya -- an analytic continuation into some region $G = G(k;\xi)\subset\C$ containing the real axis. Hence, for each $x\in\R$, we may choose $r(x,k;\xi)>0$, such that $B_{r(x,k;\xi)}^\C(x)\subset G(k;\xi)$. 
By possible decreasing $r(x,k,\xi)$, we may assume that $\gamma^z_\xi(k)\in S^d_{\tR}$ for $|z-x|< r(x,k;\xi)$. Therefore, $\gamma_\xi^z(k)-\xi \in S_{2\tR}^d$ and by Condition \ref{cond_ex}.\ref{cond_ex_v_omega} and the definition of $v_\xi$ we may form $v_\xi(\gamma_\xi^z(k))$.

For each $x\in\R$, (and  $k\in\R^d$, $\xi\in S_\tR^d$), the function $y_x(t) = \gamma^{x+\ci t}_\xi(k)$ solves the ODE
\[
\frac{\mathrm{d} y_x }{\mathrm{d}t}(t) = \ci v_\xi(y_x(t))
\]
with the initial condition $y_x(0) = \gamma^x_\xi(k)$.
The solution is a priori defined for $|t|< r(x,k;\xi)$.

The estimate (for $t>0$)
\begin{equation}\label{EstOnImPart}
|\im(y_x(t))| = |\Im(y_x(t)-y_x(0))|\leq \int_0^t \bigl|v_\xi(y_x(s))\bigr|\,\mathrm{d} s \leq C_\omega t,
\end{equation}
ensures that the solution may be extended beyond $|t| = r(x,k;\xi)$ at least until $|t|= r= \tR/(C_\omega+1)$. (A similar estimate holds for $t<0$). This defines an extension of $\gamma^z_\xi(k)$ from $z\in G(k,\xi)$ to $z\in S_{r}$ for all $k\in \R^d $, $\xi\in S_\tR^d$. 
Note that $\gamma^z_\xi(k)\in S^d_{C_\omega r}\subset S^d_\tR$ for all $z\in S_r$. It remains to argue that the extension is analytic.

By Cauchy-Kowaleskaya, for each $x\in\R$, $k\in\R^d$ and $\xi\in S_\tR^d$, the solution $t\mapsto y_x(t)$ extends to an analytic function in a complex neighborhood $O(x,k;\xi)$ of $[-r,r]$. Let $\delta = \delta(x,k;\xi) \in (0,r(x,k;\xi))$ be such that
$[-r,r]\times \ci[-\delta,\delta]\subset O(x,k;\xi)$. 

By Unique Continuation, the extension $y_x(z)$ equals $\gamma^{x+\ci z}_\xi(k)$ for $z\in ([-r,r]\times \ci[-\delta,\delta]) \cap  B^\C_{r(x,k;\xi)}(0)$.
This implies that $t\to y_x( t+\ci x')$, for $|x'|<\delta(x,k;\xi)$,  solves the same initial value problem as $y_{x'}(t)$ and hence they must coincide. This proves that $z\mapsto \gamma^z_\xi(k)$ is analytic in $(x-\delta(x,k;\xi),x+\delta(x;k;\xi))\times \ci(-r,r)$. Since $x\in \R$ was arbitrary, we conclude that  $z\mapsto \gamma^z_\xi(k)$ is analytic in $S_r$. 

The last claim about the range of $z\mapsto \gamma^z_\xi(k)$ for $|\im z|\leq r'\leq r$, follows from \eqref{EstOnImPart}.
\end{proof}

\begin{rem}
%\leavevmode
\begin{enumerate}
\item
Let $z\in S_r$ and write $u = z/|z|$. 
The above lemma allows us to estimate
\begin{equation}\label{eq_gamma_bnd}
\bigl|\gamma^{z}_\xi(k)-k\bigr| = \Bigl|\int_0^{|z|} \frac{\mathrm{d}}{\mathrm{d} r} \bigl(\gamma^{r u}_\xi(k)-k \bigr)\, d r\Bigr|
= \Bigl|\int_0^{|z|} u v_\xi\bigl(\gamma^{r u}_\xi(k)\bigr)\, d r\Bigr|
\leq C_\omega |z|.
\end{equation}
\item
Let $k,k',\xi\in\R^d$. Abbreviate $\beta^z_\xi(k,k') = \gamma^z_\xi(k')-\gamma^z_\xi(k)$ 
and $u(z) = |\beta^z_\xi(k,k')|^2$ for $z\in S_r$. Then
\[
\frac{\mathrm{d} u (z)}{\mathrm{d} z} = 2 \re\Bigl\{ \overline{\beta^z_\xi(k,k')}\cdot \frac{\mathrm{d}\beta^z_\xi(k,k')}{\mathrm{d} z}\Bigr\} =   2 \re \Bigl\{\overline{\beta^z_\xi(k,k')}\cdot \bigl(v_\xi(\gamma^z_\xi(k'))-v_\xi(\gamma^z_\xi(k))\bigr)\Bigr\}.
\]
Estimating $|v_\xi(\gamma^z_\xi(k'))-v_\xi(\gamma^z_\xi(k))|\leq C'_\omega|\beta^z_\xi(k,k')|$, using \eqref{eq_bnd_on_Dv}, we arrive at the differential inequalities
$- 2 C'_\omega u\leq \dot{u} \leq  2 C'_\omega u$. Therefore, we may conclude the estimate
\begin{equation}\label{LowerBoundOnBeta}
\forall z\in S_r:\qquad |\beta^z_\xi(k,k')|\geq |k-k'| \mathrm{e}^{-C'_\omega|z|}.
\end{equation}
\end{enumerate}
\end{rem}

The previous two lemmata allow us to explicitly compute how conjugation by the unitary group generated by $A$ effects the fiber Hamiltonians and argue that the so obtained expressions admit analytic continuations.

\begin{lem}\label{lem_int_strip}
  Assume Condition~\ref{cond_ex}. Then the map
  \begin{equation*}
   \R\ni  t\mapsto \mathrm{e}^{\ci t A_\xi}T_V \mathrm{e}^{-\ci t A_\xi}:=T_V^t
  \end{equation*}
  extends to an analytic $B(\mathcal{H})$-valued function defined in $S_{R}$, where
  \begin{equation}\label{rprime}
  R = \min\Bigl\{r,\frac{a'}{C_\omega+1},\frac{\pi}{d C'_\omega+1} \Bigr\}.
  \end{equation} 
  We furthermore have the estimate $\|T_V^z\|\leq C_V C_{d} \mathrm{e}^{(d+d')C'_\omega|z|}$ for $z\in S_{R}$,
  where  $C_d = \int_{\R^d} (1+|k|^{d'})^{-1}\, dk$.
\end{lem}
%\begin{rem}
%  We are in fact able to prove the above under the weaker assumption
%  $\norm{\partial^\gamma \hat V}_1\le c^{\abs{\gamma}}\gamma!$ for
%  some $c>0$ and all multiindices $\gamma$. The proof of this,
%  however, is much longer and more involved than the short one given
%  below.
%\end{rem}
\begin{proof}[\textsc{Proof}]
We begin by computing for $t\in \R$
\begin{align}\label{DilatedTV}
\nonumber \left(\mathrm{e}^{\ci A_\xi t} T_V \mathrm{e}^{-\ci A_\xi t}g\right)(k) & =  \sqrt{J^t_\xi(k)} \int_{\R^d}\widehat{V}\bigl(k'-\gamma^{t}_\xi(k)\bigr)\sqrt{J^{-t}_\xi(k')}g(\gamma^{-t}_\xi(k'))\,\mathrm{d}k'\\
\nonumber & =  \sqrt{J^t_\xi(k)} \int_{\R^d}\widehat{V}\bigl(\gamma^t_\xi(k'')-\gamma^{t}_\xi(k)\bigr)
\sqrt{J^{-t}_\xi(\gamma^t_\xi(k''))} J^t_\xi(k'')g(k'')\,\mathrm{d}k''\\
 &= \sqrt{J^t_\xi(k)} \int_{\R^d}\widehat{V}\bigl(\gamma_{\xi}^t(k'')-\gamma_{\xi}^t(k)\bigr)\sqrt{J^t_\xi(k'')}g(k'')\,\mathrm{d}k''.
\end{align}
Here we used substitution $k' = \gamma^t_\xi(k'')$ in the second equality and the identity \eqref{eq_J_inv} in the final equality.

By Condition~\ref{cond_ex}.\ref{cond_ex_V} and Lemma~\ref{lem_flow_strip}, the map $t\mapsto\widehat{V}(\gamma_\xi^{t}(k')-\gamma_\xi^{t}(k))$ extends to an analytic function on the strip $S_{R}$.
%As in the abstract case, we start with an extension to $|z|<r$ and then further continue to the whole strip $S_r$ by %conjugation with the\fxfatal{not what we do}
% unitary group.
 
  It follows from Condition~\ref{cond_ex}.\ref{item-cond_ex-1}, \eqref{eq_flow_J} and Lemma~\ref{lem_flow_strip} that $t\mapsto J^t_\xi(k)$ extends analytically to $z\in S_r$. The estimate \eqref{eq_bnd_on_Dv} implies:
  \begin{equation}\label{BoundOnJ}
  \forall z\in S_r,k,\xi\in \R^d: \quad |J^z_\xi(k)| \leq \mathrm{e}^{dC'_\omega|z|} \quad \textup{and} \quad |\arg(J^z_\xi(k))| \leq  dC'_\omega|\im z|,
  \end{equation}  
  where the second estimate -- on $\im(\int_0^z \nabla\cdot v_\xi(\gamma^s_\xi(k))\, \mathrm{d} s)$ -- is most easily derived by choosing a piecewise linear integration contour from $0$ to $z$, which runs along the real axis from $0$ to $\re(z)$ (contributing nothing to the argument), and then from $\re(z)$ to $z$.
Here $\arg(\zeta)$ denotes the principal argument of the complex number $\zeta$. The estimate on the argument shows that
$\sqrt{J^t_\xi(k)}$ extends analytically to $S_{R}$ (reading the square root as the principal square root).

It remains to show that, if these analytic extensions are substituted into the right-hand side of \eqref{DilatedTV}, it still defines a bounded operator on $\mathrm{L}^2(\R^d)$.
% Let $n\in\N$ denote the smallest integer satisfying $2n>d$. Then $j_n(k):=(1+|k|^{2n})\inv \in \mathrm{L}^1(\R^d)$.
%Note that due to
%\begin{equation*}
%\Bigl|\int_0^z\nabla_k\cdot v_\xi(\gamma_\xi^{z'}(k))\,\mathrm{d}z'\Bigr| \leq |z| \int_0^1\bigl|\nabla_k\cdot %v_\xi(\gamma^{tz}_\xi(k))\bigr|\, \mathrm{d}t
%\end{equation*}
%there exists a constant $C_r>0$ independent of $k$ and $|z|<r$ such that $|J^{z}_\xi(k)|\leq C_r$. 
We recall the notation  $\beta_\xi^z(k,k'):= \gamma_\xi^{z}(k')-\gamma_\xi^{z}(k)$ and estimate using \eqref{DecayOfV}, \eqref{LowerBoundOnBeta} and \eqref{BoundOnJ}
\begin{align}
&
\int_{\R^d}\bigl| J^{z}_\xi(k)\bigr| \Bigl| \int_{\R^d} \widehat{V}(\beta_{z}(k,k'))\sqrt{J^{z}_\xi(k')}g(k')\, \mathrm{d}k'\Bigr|^2 \mathrm{d}k\nonumber\\
\nonumber & \quad \leq
C_V^2 \mathrm{e}^{2 d C'_\omega|z|} \int_{\R^d} \Bigl( \int_{\R^d} \bigl(1+ |\beta_\xi^{z}(k,k')|^{d'}\bigr)^{-1}\bigl|g(k')\bigr| \,\mathrm{d}k'\Bigr)^2 \mathrm{d}k\\
\nonumber &\quad \leq C_V^2 \mathrm{e}^{2(d+d') C'_\omega|z|} \int_{\R^d} \Bigl( \int_{\R^d} \bigl(1+|k-k'|^{d'}\bigr)^{-1} \bigl|g(k')\bigr|\,\mathrm{d}k'\Bigr)^2 \mathrm{d}k\\
&\quad \leq C_V^2 C_{d}^2 \mathrm{e}^{2(d+d')C'_\omega|z|} \,\|g\|^2_{\mathrm{L}^2(\R^d)}.
\label{eq_tv1}
\end{align}
In the last step, we used (the sharp) Young's inequality \cite[Thm. 4.2]{LL}. Here $C_d = \int_{\R^d} (1+|k|^{d'})^{-1}\, dk$.
The above estimate shows that the right-hand side of \eqref{DilatedTV} defines a bounded operator
for $z\in S_{R}$.

 That the extension to complex $z\in S_{R}$ is analytic now follows from Morera's Theorem.
\end{proof}

%\begin{lem}
%Assume Condition~\ref{cond_ex}. For all $t\in\mdsr$ we have that 
%\begin{equation*}
%\mathrm{e}^{\ci t A_\xi}M_{\omega_\xi}\mathrm{e}^{-\ci t A_\xi} = %M_{\omega_\xi \circ \gamma_{t}}
%\end{equation*}
%on $D(M_{\omega_\xi})$, where $\gamma_t$ is the solution of %(\ref{eq_flow_ODE}).
%\end{lem}

%\begin{proof}[\textsc{Proof.}] Let $\psi\in D(M_{\omega_\xi})$. %Since the action of $\mathrm{e}^{\ci t A_\xi}$ is given by %(\ref{eq_action_dilation}), we use (\ref{eq_J_inv}) and the group %property of the flow $\gamma_t$ to compute 
%\begin{align*}
%\bigl[\mathrm{e}^{\ci t A_\xi}M_{\omega_\xi}\mathrm{e}^{-\ci t %A_\xi}f\bigr](k) 
%&= 
%\omega_\xi\bigl(\gamma_{t}(k)\bigr)\sqrt{J(t,k)}\sqrt{J(-t,\gamma_ %{t}(k))}  f(k)\\
%& =
%\omega_\xi\bigl(\gamma_{t}(k)\bigr)f(k).
%\end{align*}
%This proves the statement. \end{proof}

\begin{prp}\label{prp_cond21_ex}
Assume Condition \ref{cond_ex} and let $\xi,\xi_0\in \R^d$. Then the pair of operators $H(\xi)$, $A_{\xi_0}$ satisfies Condition~\ref{cond_g}
with $R$ given by \eqref{rprime} and $M=\widetilde{M}\langle\xi\rangle^{s_1}$ for some $\widetilde{M}$, which does not depend on $\xi$.
\end{prp}

\begin{rem}
  Below, the proposition above is proven by checking each of the four
  assumptions in Condition~\ref{cond_g} directly. In view of
  Proposition~\ref{prp_equiv}, it can also be proven by proving
  Condition~\ref{cond_g}.\ref{cond_g_mourre} and the
  commutator bound in
  Proposition~\ref{prp_equiv}.\ref{item:equiv-cond-2}. In fact, two of
  the authors have persued this idea, see \cite{ER}. However, their
  proof is much longer than the one presented here.
\end{rem}
\begin{proof}[\textsc{Proof of Proposition~\ref{prp_cond21_ex}.}] We begin by establishing Condition \ref{cond_g}.\ref{cond_g_mourre}. 
Since $D(H(\xi))=  \dom = D(M_{\langle k\rangle^{s_2}})$, it suffices to show that
$\mathrm{e}^{\ci t A_{\xi_0}}\colon  \dom \to  \dom$ for $t\in [0,1]$ and that
\begin{equation}
\forall\psi\in \dom\colon\qquad \sup_{t\in [0,1]}\|M_{\langle k \rangle^{s_2}} \mathrm{e}^{\ci t A_{\xi_0}}\psi\| <\infty.
\label{eq_cond1}
\end{equation}
But this follows immediately from \eqref{Peetre}, \eqref{FinPropSpeed} and Lemma~\ref {DilatedMultOps}. 
 
As for Condition \ref{cond_g}.\ref{cond_g_C1A}, we abbreviate 
 $w_0:=\frac{1}{2} \nabla\cdot v_{\xi_0}$, note that 
 \[
 A_{\xi_0} = \ci \nabla_k \cdot v_{\xi_0} - \ci w_0 = \ci v_{\xi_0}\cdot \nabla_k + \ci w_0  
 \]
  and
compute as a commutator form on $C_0^\infty(\R^d)$:
\begin{equation}\label{eq_comm}
\begin{aligned}
H(\xi)A_{\xi_0} -  A_{\xi_0}H(\xi)  
& = -\ci  v_{\xi_0}\cdot\nabla_k\omega_\xi -\ci M_{w_0}T_{V} -\ci  T_{V}M_{w_0}\\
&\quad - \ci 
\sum_{\sigma=1}^d\bigl(T_{\ci x_\sigma V}M_{(v_{\xi_0})_\sigma}- M_{(v_{\xi_0})_\sigma}T_{\ci x_\sigma V}\bigr).
\end{aligned}
\end{equation}
The right-hand side defines a bounded operator.  Condition \ref{cond_g}.\ref{cond_g_C1A} now follows from
\cite[Prop.~II.1]{Mo}, since  $\mathrm{e}^{\ci t A_{\xi_0}} C_0^\infty(\R^d)\subset C_0^\infty(\R^d)$.

By Lemma~\ref{DilatedMultOps}, Lemma \ref{lem_flow_strip} and Lemma \ref{lem_int_strip}, we may conclude that the map $t\mapsto H_t(\xi)\psi$, for $\psi\in D(H(\xi))$, admits an analytic continuation into the strip $S_{R}$. Here $R$ is defined in \eqref{rprime}. The extension is given by
$H_z(\xi) = M_{\omega_\xi \circ \gamma_{\xi_0}^z} + T_V^z$. This establishes Condition \ref{cond_g}.\ref{cond_g_ext}.

It thus remains to examine whether or not Condition \ref{cond_g}.\ref{cond_g_uniform_bd} holds. Using  Lemma~\ref{lem_int_strip}, we may estimate for $\psi\in \cH$ and $z\in S_{R}$ : 
%(\ref{eq_cond3}) shows that the operator norm of the analytic continuation of $\mathrm{e}^{\ci t %A_\xi}T_V\mathrm{e}^{-\ci tA_\xi}$ is uniformly bounded on $B_r^{\mdsc}(0)$. Denote this upper bound by $C_V>0$. %Substituting $C_V$ for $\|T_V\|$ in (\ref{eq_cond1}) and then use (\ref{eq_cond2}) we obtain for $z\in S_{r'}$ the %estimate
\begin{align}
&
\|H_z(\xi)(H(\xi)+\ci)\inv\psi\| \nonumber\\
& \quad \leq 
C_V C_d\mathrm{e}^{(d+d')C'_\omega|z|}\|(H(\xi)+\ci)\inv\psi\| + \|M_{\omega_\xi\circ\gamma_{\xi_0}^z}(H(\xi)+\ci)\inv\psi\|\nonumber\\
& \quad \leq C_V C_d\mathrm{e}^{(d+d')C'_\omega|z|}\|\psi\| + C_\omega(2^{s_1}\langle\xi\rangle^{s_1}+1) 2^{s_2} \langle C_\omega z\rangle^{s_2}
\| M_{\langle k\rangle^{s_2}}(H(\xi)+\ci)^{-1}\psi\|.
\label{eq_const_d_xi}
\end{align}
Here we used that $s_2\geq s_1$, Condition~\ref{cond_ex}.\ref{item-cond_ex2} (from where the constant $C_\omega$ comes), \eqref{Peetre} and \eqref{eq_gamma_bnd} to estimate 
\begin{align*}
|\omega_\xi(\gamma_{\xi_0}^z(k))| & \leq C_\omega\langle \xi-\gamma_{\xi_0}^z(k)\rangle^{s_1} + C_\omega\langle \gamma_{\xi_0}^z(k)\rangle^{s_2}\\
&  \leq C_\omega(2^{s_1}\langle\xi\rangle^{s_1}+1)\langle \gamma_{\xi_0}^z(k)\rangle^{s_2}\\
&\leq  C_\omega(2^{s_1}\langle\xi\rangle^{s_1}+1) 2^{s_2} \langle C_\omega  z\rangle^{s_2}  \langle k\rangle^{s_2}.
\end{align*}
The proposition now follows since $D(H(\xi)) = D(M_{\langle k \rangle^{s_2}})$.
 \end{proof}

\begin{prp}\label{prop-ME} Assume  Condition \ref{cond_ex}. Let $(\lambda,\xi)\in \R\times\R^d\setminus \mathcal T$.
Then there exist $e,\kappa,C>0$ and a compact self-adjoint operator $K$, such that
\[
\ci [H(\xi),A_\xi] \geq e - CE_{H(\xi)}(\R\setminus [\lambda-\kappa,\lambda+\kappa]) - K,
\] 
in the sense of forms on $D(H(\xi))$. If $(\lambda,\xi)\not\in \Sigma_\mathrm{pp}$, then one may choose $K=0$.
\end{prp}

\begin{proof} Note first that it suffices to establish the form estimate on $C_0^\infty(\R^d)$. 

  We abbreviate
  $g_{\xi}(k):=\mathrm{e}^{-k^2-\xi^2}|\nabla_k\omega_{\xi}(k)|^2$ 
  and write
 \begin{equation}\label{eq_comm2}
    \ci[A_{\xi},H({\xi})] 
    = 
    M_{g_{\xi}} + K',
  \end{equation}
   where $K'$ is $\ci$ times the sum of the terms on the right-hand side in \eqref{eq_comm} involving $V$.
   
   Using the $\lambda\not\in \mathcal T(\xi)$ and that $\mathcal T(\xi)$ is closed, there exists $\kappa>0$ such that
   $[\lambda-2\kappa,\lambda+2\kappa]\subset \R\setminus \mathcal T(\xi)$.
   Put 
   \[
   e := \inf\bigl\{g_{\xi}(k)\, \big| \, |\, k\in\R^d \textup{ with } |\omega_{\xi}(k)-\lambda|\leq 2\kappa\bigr\}>0
   \]
   and $C:= \sup \{|g_{\xi}(k)\,|\, k\in\R^d \}<\infty$.

    Choose an  $f\in C_0^\infty(\R)$ satisfying: $\mathrm{supp}(f) \subset [\lambda-2\kappa,\lambda+2\kappa]$,
     $0\leq f\leq 1$, and $f(\lambda')=1$ for
     $|\lambda'-\lambda|\leq \kappa$. 
     
     We may now estimate using the just chosen $e,C$ and $f$ :
      \begin{align}\label{EstOnMg}
      \nonumber   M_{g_{\xi}} &= M_{g_{\xi}}f(M_{\omega_{\xi}})^2 + M_{g_{\xi}}(1-f^2(M_{\omega_{\xi}}))\\
     \nonumber   &\geq e  f(M_{\omega_{\xi}})^2\\
        & \geq e - 2 e(1-f(M_{\omega_{\xi}})). 
     \end{align}

   To pass from $1-f(M_{\omega_{\xi}})$ to $1-f(H(\xi))$ we compute
    \begin{align*}
      f(H(\xi))-f(M_{\omega_{\xi}})
      &=\frac{1}{\pi}\int_\C\bar\partial_z\tilde f(z)((H(\xi)-z)\inv-(M_{\omega_{\xi}}-z)\inv)\,\mathrm{d} z\\
      &=\frac{1}{\pi}\int_\C\bar\partial_z\tilde
      f(z)(H(\xi)-z)\inv T_V (M_{\omega_{\xi}}-z)\inv\,\mathrm{d} z.
    \end{align*}
    Here $\tilde f$ is an almost analytic extension of $f$.
    The operator $T_V (M_{\omega_{\xi}}-z)\inv$ is  compact (for the same reason $K'$ was compact).
    This shows that $K''= f(H(\xi))-f(M_{\omega_{\xi}})$ is a compact operator.
    
    We may now combine \eqref{eq_comm2}, \eqref{EstOnMg} and the estimate $1-f(H(\xi))\leq E_{H(\xi)}(\R\setminus [\lambda-\kappa,\lambda+\kappa])$ to arrive at the Mourre estimate
    \[
     \ci[A_{\xi},H({\xi})] \geq e - 2 e E_{H(\xi)}(\R\setminus [\lambda-\kappa,\lambda+\kappa])
     - K, 
    \]
    where $K = - K' +2 e K''$ is compact.
    
    Finally, if $\lambda\not\in \Sigma_\mathrm{pp}(\xi)$, then $K E_{H(\xi)}([\lambda-\kappa,\lambda+\kappa])\to 0$ for $\kappa\to 0$ in operator norm. This implies the remaining claim.
\end{proof}

\begin{prp}\label{prop-ex-is-ok}
  Assume Condition \ref{cond_ex}. Suppose $(\lambda_0,\xi_0)\in \Sigma_\mathrm{pp}\setminus \mathcal T$.
 Let $U:=B_{\tR}^{\R^d}(\xi_0)$ and $U_{\C}:=B_{\tR}^{\C^d}(\xi_0)$. Then  Condition~\ref{cond_analytic} is satisfied with $A = A_{\xi_0}$ and $M$ replaced by $M\sup_{\xi\in U} \langle\xi\rangle^{s_1}$.
\end{prp}
%Note that the threshold set $\mathcal{T}_\xi$ is locally finite and closed, due to $|\omega_{\xi_0}(k)|\rightarrow %\infty$ as $|k|\rightarrow\infty$ 
\begin{proof}
  Let $\xi\in U$ and put $U(t):=\mathrm{exp}(\ci A_{\xi_0}t)$ for $t\in\R$. That Condition~\ref{cond_analytic}.\ref{item-AnaCond-1} holds, that is that $D(H(\xi))=:\dom$ is independent of $\xi$,  was discussed in Remark~\ref{rem-3.2}.\ref{item-rem3.2.1}.
  
 Condition~\ref{cond_analytic}.\ref{item-AnaCond-2}
  follows from Proposition~\ref{prp_cond21_ex} (with $M = \sup_{\xi\in U}\langle \xi\rangle^{s_1} \widetilde{M}$ and the $R$ from \eqref{rprime}). Hence it remains to establish Condition~\ref{cond_analytic}.\ref{item-AnaCond-3} and~\ref{cond_g_fam} (with the same $R$).

  Condition~\ref{cond_Mourre_estimateNEW}.\ref{item:Cond-ev} is satisfied by assumption, and
 Condition~\ref{cond_Mourre_estimateNEW}.\ref{item:CondMourre} follows from Proposition~\ref{prop-ME}, since 
 $(\lambda_0,\xi_0)\not\in \mathcal T$ and $\langle H(\xi_0)\rangle \geq 1$. This verifies Condition~\ref{cond_analytic}.\ref{item-AnaCond-3}.

      Put
      \[
      r_0 =  \frac{R}{1+C_\omega},
      \]
       where $R>0$ is as in \eqref{rprime} and the constants $C_\omega$ and $C_\omega'$ are from \eqref{eq_bnd_on_v}
       and \eqref{eq_bnd_on_Dv}, respectively.
     In order to verify Condition~\ref{cond_analytic}.\ref{cond_g_fam},
      choose $z_0\in B_{r_0}^{\C}(0)\subset B_R^{\C}(0)$ with $\Im(z_0)>0$.

      Due to Lemma \ref{lem_flow_strip}, $\gamma_\xi^{z_0}(k) \in S_{r_0 C_\omega}^d \subset S_R^d\subset S_{\tilde{R}}^d$ for $k\in\R^d$ and $\xi\in U_\C\subset S_{\tR}^d$. 
      We therefore have $\xi-\gamma_\xi^{z_0}(k)\in S_{2\tR}^d$
      for all $\xi\in U_\C$ and $k\in \R^d$.
      
       By Condition~\ref{cond_ex}.\ref{item-cond_ex-1}, the map
      $U\ni \xi\to  M_{\omega_\xi\circ \gamma^{z_0}_\xi} \psi$ extends to an analytic function in $U_\C$ for every $\psi\in\dom$. To see that $\xi \to T_V^{z_0}\psi$ extends from $U$ to an analytic function defined for
      $\xi\in U_\C$, we observe from \eqref{DilatedTV} (with $t$ replaced by $z_0$ on the right-hand side) that it suffices to ensure that $|\arg(J_\xi^{z_0}(k))|<\pi$ for all
      $k\in \R^d$ and $\xi\in U_\C$. But this follows from \eqref{eq_bnd_on_Dv}, \eqref{eq_flow_J} and the estimate
      $|\int_0^{z_0} \nabla\cdot v_\xi(\gamma_\xi^s(k))\, \mathrm{d} s| \leq |z_0| d C_\omega'$, valid for all
      $k\in \R^d$ and $\xi\in S_\tR^d$. Here we used that $r_0\leq R \leq \pi/(1+dC_\omega')$, cf.~\eqref{rprime}. 
      
      It follows that $U_{\C}\ni\xi\mapsto H_{z_0}(\xi)\psi$ is analytic for all $\psi\in\dom$.
      This verifies Condition~\ref{cond_analytic}.\ref{cond_g_fam}, since all operators have a common domain by Proposition \ref{prp_cond21_ex}.            
%        Let $k\in \C^d, k'\in S_R^d$ with $k+k'\in S_{2R}^d$. We use Condition \ref{cond_ex}.\ref{cond_ex_omega_est} to estimate
%        \begin{align*}
%        |\omega_1(k+k')| \leq \tilde{C} \langle k+k'\rangle^{s_1} \leq \tilde{C}2^{s_1}\langle k \rangle^{s_1} \langle k'\rangle^{s_1} \leq 2^{s_1}\langle k \rangle^{s_1}(|\omega_1(k')| + \tilde{C}^2).
%        \end{align*}
%        
%        We choose $k= \xi-\xi'$ and $k'= \xi'-\gamma_{\xi_0}^{-t_0}(k)$, where $\xi,\xi'\in B_R^{\C}(0)$, and compute
%        
%        \begin{align*}
%       \left\|M_{\omega_\xi\circ\gamma_{\xi_0}^{-t_0}}\psi \right\|^2 
%       &\leq 
%       \left[ \int_{\R^d}|\omega_2(\gamma_{\xi_0}^{-t_0}(k))|^2|\psi(k)|^2\mathrm{d}k \right]^{\frac{1}{2}}\\
%       &~~~+ 
%        2^{s_1}\langle R\rangle^{s_1}\left[\int_{\R^d} (|\omega_1(\xi'-\gamma_{\xi_0}^{-t_0}(k))|+\tilde{C}^2)^2|\psi(k)|^2\mathrm{d}k\right]^{\frac{1}{2}}.
%        \end{align*}
%        Since we can reverse roles of $\xi,\xi'$, the estimate implies $D(H_{t_0}(\xi))=D(M_{\omega_\xi\circ\gamma_{\xi_0}^{-t_0}})=D(M_{\omega_{\xi'}\circ\gamma_{\xi_0}^{-t_0}})= D(H_{t_0}(\xi'))$. Thus, the family $\{H_{t_0}(\xi)\}_{\xi\in B_R^{\C}(0)}$ is of type $A$ by definition.
\end{proof}

\begin{proof}[Proof of Theorem~\ref{thm-ex}]
Let $M = \R\times\R^d$ and $N = \R^d$, both real analytic manifolds. 
To show that $\mathcal T$ is sub-analytic, we write it as the image under a real analytic map of the semi-analytic subset 
\[
Y = \bigl\{(\lambda,\xi,k)\in M\times N\, \big| \, \omega_\xi(k) -\lambda=0\big\} \cap \Bigl(\cap_{j=1}^d\bigl\{(\lambda,\xi,k)\in M\times N \, \big|\, \partial_{k_j}w_\xi(k) = 0 \bigr\}\Bigr)
\]
 of $M\times N = \R\times\R^d\times\R^d$. Here we used that the semi-analytic subsets of $M$ form a ring.
 Let $\pi\colon M\times N\to M$ be the projection $\pi(\lambda,\xi,k) = (\lambda,\xi)$. 
Then $\mathcal T = \pi(Y)$. Since $|\omega_\xi(k)| \to\infty$ as $|k|\to\infty$, this implies that $\mathcal T$ is sub-analytic ($\pi_{|Y}$ is a proper map). Since $\mathcal T$ is a closed subset of $M$, this  finishes the proof of \ref{item-thm-ex-1}.

Fix $\xi\in \R^d$.
To see that $\mathcal T(\xi)$ is locally finite, consider the map $\pi_\xi$ from 
$M_\xi = \{(\omega_\xi(k),k)\, |\, k\in \R^d\} $ into $\R$ defined by setting $\pi_\xi(\lambda,k) = \lambda$. 
Since $M_\xi$ and $\R$ are real analytic manifolds and $\pi_\xi$ is a real analytic proper map, we get
from Theorem~\ref{thm-strat} sub-analytic stratifications $\{S_\alpha\}_{\alpha\in A}$ of $M_\xi$ and $\{T_\beta\}_{\beta\in B}$ of $\R$.

Let $\lambda\in \mathcal T(\xi)$. Then there exists $k\in \R^d$, such that $\omega_\xi(k) = \lambda$ and $\nabla_k \omega_\xi(k)=0$. Let $\alpha\in A$ be such that $(\lambda,k)\in S_\alpha$. 
Since $\mathrm{d} \pi_{\xi|S_\alpha}(\lambda,k)(\eta,u) = \eta$ and $\eta = u\cdot \nabla_k\omega_\xi(k)$ for $(\eta,u)\in T_{(\lambda,k)} S_\alpha$, we conclude that $\mathrm{rank}(\mathrm{d} \pi_{\xi|S_\alpha}(\lambda,k)) = 0$.
Let $\beta\in B$ be such that $\pi_\xi(S_\alpha) = T_\beta$. Then $T_\beta$ is a zero-stratum, hence a singleton, which must be $T_\beta = \{\lambda\}$.
 By local finiteness of the stratification $\{T_\beta\}_{\beta\in B}$,  we may therefore conclude that $\mathcal T(\xi)$ is locally finite. 
 This concludes the proof of \ref{item-thm-ex-2}.

 We now turn to \ref{item-semian-ex}. In order to verify that $\Sigma_\mathrm{pp}\setminus \mathcal T$ is a (relatively) closed and semi-analytic subset of $\R\times \R^d\setminus \mathcal T$, pick $(\lambda,\xi) \in \R\times \R^d\setminus \mathcal T$. 
 
 If $(\lambda,\xi)\not\in \Sigma_\mathrm{pp}$, then we obtain from Proposition~\ref{prop-ME}
 constants $e,\kappa,C>0$ such that
 \[
 \ci [H(\xi),A_\xi]\geq e - CE_{H(\xi)}([\lambda-\kappa,\lambda+\kappa]).
 \]
 By a continuity argument, there exists a sufficiently small open neighborhood $W\subset \R\times\R^d\setminus \mathcal T$ of 
 $(\lambda,\xi)$,  perhaps new smaller positive constants $e'<e$, $\kappa'<\kappa$ and a bigger $C'>C$, such that
 \[
 \forall (\lambda',\xi')\in W: \qquad \ci [H(\xi'),A_{\xi'}]\geq e' - C'E_{H(\xi')}([\lambda'-\kappa',\lambda'+\kappa']).
 \]
 Here we used that the map $\xi'\to [H(\xi'),A_{\xi'}]\in \mathcal B(\mathrm L^2(\R^d))$ is continuous, cf. \eqref{eq_comm}. 
 Hence, by the Virial Theorem \cite{GeoGer_virial}, $\Sigma_\mathrm{pp}\cap W  =\emptyset$. This shows that $\Sigma_\mathrm{pp}\setminus \mathcal T$ is a closed subset of $\R\times\R^d\setminus \mathcal T$, and that $W\cap (\Sigma_\mathrm{pp}\setminus \mathcal T)\in\mathcal O(W)$. Recall from Definition~\ref{def-semi-sub}.\ref{item-def-OW}, the definition of $\mathcal O(W)$.  
 
 As for the case $(\lambda,\xi)\in \Sigma_\mathrm{pp}$, we obtain from Theorem~\ref{thm_semianalytic} a neighborhood $W$ of $(\lambda,\xi)$,
 such that $\Sigma_\mathrm{pp}\cap W\in\mathcal O(W)$. Here we used Proposition~\ref{prop-ex-is-ok} to ensure the applicability of Theorem~\ref{thm_semianalytic}.
This completes the proof of \ref{item-semian-ex}.

The last property \ref{item-thm-ex-4} is a standard consequence of having a Mourre estimate satisfied for energies
$\lambda\in \R\setminus \mathcal T(\xi)$.
\end{proof}

\renewcommand{\thesection}{\Alph{section}}
\renewcommand{\theequation}{\Alph{section}.\arabic{equation}}
\setcounter{section}{0}
\setcounter{equation}{0}
\section{Approximate Eigenstates for Closed Operators}

\begin{lem}\label{lem-dich} Let $T$ be a densely defined closed operator on a Hilbert space $\cH$. Let $\mu\in \sigma(T)$.
At least one of the following two properties hold true:
\begin{enumerate}
\item\label{item-dich-1} There exists a sequence $\{\psi_n\}_{n=1}^\infty \subset D(T)$ with $\|\psi_n\| = 1$ for all $n$ and
\[
\lim_{n\to\infty} \|(T-\mu)\psi_n\| = 0,
\] 
\item\label{item-dich-2} There exists a sequence $\{\phi_n\}_{n=1}^\infty \subset D(T^*)$ with $\|\phi_n\| = 1$ for all $n$ and
\[
\lim_{n\to\infty} \|(T^*-\bar{\mu})\phi_n\| = 0.
\] 
\end{enumerate}
\end{lem}

\begin{proof} Note that $\mu\in \sigma(T)$ implies that $\bar{\mu}\in\sigma(T^*)$. Furthermore, if $T-\mu$ has a bounded left-inverse $L$, then $L^*$ is a bounded right-inverse of $T^*-\bar{\mu}$. Conversely, if $L$ is a bounded left-inverse of $T^*-\bar{\mu}$, then $L^*$ is a bounded right-inverse of $T-\mu$. In particular, it is not possible for $T-\mu$ and $T-\bar{\mu}$ to both have a bounded left-inverse. Therefore we may suppose that $T-\mu$ does not have a bounded left-inverse.%\begin{proof} Suppose that $\lambda\in\rho(T)$. Then there exists a bounded inverse $B$ such that $B(T-\lambda)=1$. Thus, 
%\begin{equation*}
%\|\psi\| = \|B(T-\lambda)\psi\| \leq \|B\| \|(T-\lambda)\psi\|
%\end{equation*} 
%and $\|(T-\lambda)\psi\|\geq c\|\psi\|$ follows from $\|B\|\neq 0$.

Suppose \ref{item-dich-1} is false. Then there exists $c>0$, such that for all $\psi\in D(T)$ we have
$\|(T-\mu)\psi\|\geq c\|\psi\|$. This coercive estimate ensures that 
$T-\mu$ is injective and that $V = \Ran(T-\mu)$ is closed. If $V=\cH$, then $T-\mu$ has a bounded left inverse by the Closed Graph Theorem, which we assumed it did not have, and hence $V\subsetneq\cH$. We may now pick $\phi\in V^\perp$ with $\|\phi\| = 1$.
Then, for $\psi\in D(T)$,  $|\langle \phi, T\psi\rangle| = |\langle \phi ,\mu\psi\rangle| \leq |\mu| \|\psi\|$. This shows that
$\phi\in D(T^*)$. For any $\psi\in D(T)$, we may now compute $\langle T^*\phi,\psi\rangle = \langle \phi,T\psi\rangle = 
\langle \phi,\mu\psi\rangle = \langle \bar{\mu}\phi,\psi\rangle$, and conclude that $T^*\phi = \bar{\mu} \phi$. This concludes the proof, since \ref{item-dich-2} will hold with the constant sequence $\phi_n = \phi$ for all $n$.%Conversely, let us assume that there exists $c>0$ such that $\|(T-\lambda)\psi\|\geq c\|\psi\|$. Note that $V:=\mathrm{Ran}(T-\lambda)$ is a closed subspace of $\mathrm{H}$. Choose $\phi \perp V$ and compute 
%\begin{equation*}
%|\langle \phi, T\psi\rangle|\leq |\langle \phi, (T-\lambda)\psi\rangle| + |\lambda| |\langle \phi, \psi\rangle | \leq |\lambda|\|\phi\|\|\psi\|
%\end{equation*}
%for every $\psi\in D(T)$. This implies that $\phi \in D(T^*)$ and we can thus calculate
%\begin{equation*}
%\langle (T^*-\overline{\lambda})\phi, \psi\rangle = \langle \phi, (T-\lambda)\psi \rangle =0
%\end{equation*}
%whenever $\psi\in D(T)$. Since $D(T)$ is dense, we have now established that $\overline{\lambda}\in \sigma_{\mathrm{pp}}(T^*)$. The intertwining relation $T\mathcal{C}=\mathcal{C}T^*$ then gives the contradiction that $\lambda\in \sigma_{\mathrm{pp}}(T)$. Therefore, we must have $V=\mathcal{H}$. The preceding argument shows that $T-\lambda$ is bijective and we can conclude that it has a left inverse $B$. Since
%\begin{equation*}
%\|B(T-\lambda)\psi\| =\|\psi\|\leq \frac{1}{c}\|(T-\lambda)\psi\|,
%\end{equation*}
%$B$ is bounded. 
%Due to $\conj^2 = 1$, the anti-linearity of $\conj$ and the intertwining property, the estimate $\|(T-\lambda)\psi\| \geq c\|\psi\|$ on $D(T)$ can be re-written as
%\fxwarning{In the notes you gave me this (18) was with $\leq$. Moreover you just stated that it was true. But doesn't this use that $T$ is normal? (which you assume in the beginning) How does this relate to the case where $T$ is just closed? To connect it with our case i added the assumption $D(T)=D(T^*)$ which holds for $H_\theta(\xi)$ 
\end{proof}

\section{Semi- and Sub-analytic geometry}\label{sec-semi-sub}

In this appendix we rather briefly summarize the notions of \emph{semi-analytic} and \emph{sub-analytic} sets.
For further background, we refer the reader to \cite{BM,Del} and references therein.

\begin{defi}\label{def-semi-sub} Let $M$ be a real analytic manifold.
\begin{enumerate}
\item\label{item-def-OW} Let $W\subset M$ be an open nonempty set. We write $\mathcal O(W)$ for the smallest ring\footnote{collection of sets stable under complement as well as under finite intersections and unions.} of subsets of $W$  
containing sets of the form $\{y\in W \, | \, f(y)>0\}$ and  $\{y\in W \, | \, f(y)=0\}$ where $f$ ranges over real analytic functions $f\colon W\to \R$.
\item\label{item-def-semi} A subset $X\subset M$ is called a semi-analytic subset of $M$
if: for any $x\in M$, there exists an open neighborhood $W\subset M$ of $x$, such that $ X\cap W\in \mathcal O(W)$.
\item\label{item-def-sub} A subset  $X\subset M$ is called a sub-analytic subset of $M$ if: for each point $x\in M$,
there exists an open neighborhood $x\in W\subset M$, a real analytic manifold $N$ and a 
semi-analytic subset $Y\subset M\times N$, such that
\begin{itemize}
\item The closure $\overline{Y}$ inside $M\times N$ is compact. 
\item $\pi(Y) = X\cap W$, where $\pi\colon M\times N\to \R^n$ is the projection onto the first coordinate. 
\end{itemize}
\end{enumerate}
\end{defi} 
The semi-analytic as well as the sub-analytic subsets of $M$ form rings of subsets of $M$.  Semi-analytic subsets are of course sub-analytic as well. The converse is in general false, but if $\dim(M)\leq 2$, the two notions are the same.

\begin{defi}
Let $M$ be a real analytic manifold and $X\subset M$ a semi-analytic (sub-analytic) subset. Let $\{S_\alpha\}_{\alpha\in A}$ be a collection of subsets of $X$. We say that $\{S_\alpha\}_{\alpha\in A}$ is a \emph{semi-analytic stratification} (\emph{sub-analytic stratification}) if
\begin{itemize}
\item Each $S_\alpha$ is a connected real analytic manifold.
\item $\cup_{\alpha\in A} S_\alpha = X$ and $S_\alpha\cap S_\beta = \emptyset$ for $\alpha\neq \beta$.
\item For any $K\subset M$ compact, the set $\{\alpha\in A\,|\, S_\alpha\cap K\neq \emptyset\}$ is finite. (\emph{Local finiteness}.)
\item If $\alpha\neq \beta$ and $\overline{S}_\alpha\cap S_\beta\neq \emptyset$, then $S_\beta\subset \partial S_\alpha$. (\emph{Frontier condition}.)
\item Each $S_\alpha$ is semi-analytic (sub-analytic) as a subset of $M$. 
\end{itemize} 
\end{defi}
 The sets $S_\alpha$ are called strata, more specifically $k$-strata if $\dim(S_\alpha) = k$.
 We note that any semi-analytic (sub-analytic) subset $X\subset M$ admits a semi-analytic (sub-analytic) stratification.

 In this paper we need the following result: 
 \begin{thm}\label{thm-strat}
 Let $M$ and $N$ be real analytic manifolds and $\pi \colon M\to N$
 a proper real analytic map.
 Then there exist sub-analytic stratifications $\{S_\alpha\}_{\alpha \in A}$ of $M$ and $\{T_\beta\}_{\beta\in B}$ of $N$, such that
 for each $\alpha\in A$, there exists $\beta\in B$ with
 \begin{itemize}
 \item $\pi(S_\alpha) = T_\beta$.
 \item $\mathrm{rank}(\mathrm{d}\pi_{|S_\alpha}(s)) = \dim(T_\beta)$ for all $s\in S_\alpha$.
 \end{itemize}
\end{thm}

\bibliographystyle{amsplain}
\bibliography{library}

\listoffixmes

\end{document}